\documentclass[nohyperref]{article}

\usepackage{microtype}
\usepackage{graphicx}
\usepackage{booktabs} 

\usepackage{hyperref}

\usepackage[accepted]{icml2022}

\usepackage{amsmath}
\usepackage{amssymb}
\usepackage{mathtools}
\usepackage{bbold}
\usepackage{amsthm}
\usepackage{float}
\usepackage{subcaption}

\usepackage[capitalize,noabbrev]{cleveref}

\theoremstyle{plain}
\newtheorem{theorem}{Theorem}[section]

\newtheorem{lemma}[theorem]{Lemma}
\newtheorem{corollary}[theorem]{Corollary}
\newtheorem{property}{Property}
\theoremstyle{definition}
\newtheorem{definition}[theorem]{Definition}

\theoremstyle{remark}

\usepackage[textsize=tiny]{todonotes}

\newcommand{\xx}{\ensuremath{\boldsymbol{{x}}}}
\newcommand{\xxi}{\ensuremath{\boldsymbol{\xi}}}
\newcommand{\zz}{\ensuremath{\boldsymbol{z}}}
\newcommand{\hh}{\ensuremath{\boldsymbol{h}}}

\newcommand{\XX}{\ensuremath{\boldsymbol{X}}}
\newcommand{\pp}{\ensuremath{\boldsymbol{p}}}
\newcommand{\bva}{\ensuremath{\boldsymbol{a}}}
\newcommand{\bvv}{\ensuremath{\boldsymbol{v}}}

\newcommand{\mean}{\ensuremath{\boldsymbol{\mu}}}
\newcommand{\varm}{\ensuremath{\boldsymbol{\Sigma}}}

\newcommand{\RR}{\ensuremath{\boldsymbol{R}}}

\newcommand\norma[1]{\left\lVert#1\right\rVert}

\icmltitlerunning{DiffSim for Rare Events}

\begin{document}

\twocolumn[
\icmltitle{Differentiable Simulations for Enhanced Sampling of Rare Events }



\begin{icmlauthorlist}

\icmlauthor{Martin Šípka}{MIT,MFF}
\icmlauthor{Johannes C. B. Dietschreit}{MIT}
\icmlauthor{Lukáš Grajciar}{PRF}
\icmlauthor{Rafael Gómez-Bombarelli}{MIT}

\end{icmlauthorlist}

\icmlaffiliation{MIT}{Department of Materials Science and Engineering, Massachusetts Institute of Technology, Cambridge, Massachusetts 02139, USA}
\icmlaffiliation{MFF}{Mathematical Institute, Faculty of Mathematics and Physics, Charles University, Sokolovská 83, 186 75 Prague, Czech Republic}
\icmlaffiliation{PRF}{Department of Physical and Macromolecular Chemistry, Faculty of Sciences, Charles University, 128 43
Prague 2, Czech Republic}

\icmlcorrespondingauthor{Rafael Gómez-Bombarelli}{rafagb@mit.edu}

\icmlkeywords{Differentiable simulations, enhanced sampling, neural differential equations}

\vskip 0.3in
]



\printAffiliationsAndNotice{}  

\begin{abstract}
Simulating rare events, such as the transformation of a reactant into a product in a chemical reaction typically requires enhanced sampling techniques that rely on heuristically chosen collective variables (CVs). 
We propose using differentiable simulations (DiffSim) for the discovery and enhanced sampling of chemical transformations without a need to resort to preselected CVs, using only a distance metric. 
Reaction path discovery and estimation of the biasing potential that enhances the sampling are merged into a single end-to-end problem that is solved by path-integral optimization. 
This is achieved by introducing multiple improvements over standard DiffSim such as partial backpropagation and graph mini-batching making DiffSim training stable and efficient.
The potential of DiffSim is demonstrated in the successful discovery of transition paths for the Muller-Brown model potential as well as a benchmark chemical system - alanine dipeptide.
\end{abstract}

\section{Introduction}
\label{intro}

A chemical reaction can be viewed as a transition from one depression (reactant) on the potential energy surface (PES) to another (product). 
The most likely transition path(s) connecting the two basins define the reaction mechanism(s). 
The potential energy of a saddle point, through which the system has to pass, defines the reaction barrier and is the fundamental quantity when investigating reaction rates.
The major obstacle in determining the reaction path lies in the high dimensionality of the molecular configuration space that can easily be spanned by thousands of degrees of freedom (DoF). 
Extensive sampling of configurations along candidate transition paths, characterized by comparatively high free energies\cite{Chipot2007, Chipot2014a} is needed, but standard unbiased sampling algorithms, e.g., molecular dynamics (MD) or Monte-Carlo (MC), often remain trapped in (meta)stable regions.
Therefore, it is extremely inefficient to explore candidate paths in an unbiased way, and it is necessary to adopt heuristics to bias the exploration, which are often based on expert chemical intuition. 

The problem has been commonly split into two seemingly easier sub-tasks. 
First, a dimensionality reduction from all DoFs down to the so-called collective variables (CVs) and second, enhanced sampling along those CVs \cite{Torrie1977a, Darve2001a, Laio2002a, Abrams2013a, Spiwok2015, Valsson2016}. 
Even though widely used methods exist to solve the second problem, the first part - identifying collective variables - is still largely a manual task based on expert chemical intuition, with the commonly used CVs being not much more complex than simple linear combinations of manually chosen internal DoFs of the molecular systems in question. 
Recently, the task of identifying CVs has been partially automatized by multiple machine learning based tools.\cite{Sultan2018AutomatedLearning, Mendels2018b, Wehmeyer2018Time-laggedKinetics, Wang2019e, Bonati2020a, Wang2021a, Sun2022MultitaskEvents, Sipka2022UnderstandingRepresentations} 
The number of the CVs is typically limited to one to three due to the exponential growth of computational cost, known as the curse of dimensionality \cite{Bellman1967, Koppen}.
The CVs should be based on those DoFs, which fully describe the rare transition event, and are thereby associated with the slowest motions. 
However, identification of the important DoFs \textit{a priori} typically requires knowledge of the transition path that one is trying to discover in the first place, \textit{i.e.}, one still ends up with the proverbial "chicken-and-egg problem" \cite{Rohrdanz2013DiscoveringReactions}. This is a problem that previously proposed machine learning based tools cannot directly tackle. 
Additionally, once the CVs are chosen it is very difficult to correct them on-the-fly. 
Thus, one must be certain that the chosen CV function is properly defined and well behaved in all regions (i.e., the CV values for reactant and product basins do not overlap or it shows undefined behavior for unseen configurations). 
A hard task for tools such as neural networks as they often extrapolate poorly when presented with unseen data.
To solve these problems, iterative improvements of CVs have been proposed \cite{Chen2018CollectiveDesign, Belkacemi2022ChasingTrajectories} where both, CV training algorithm and biasing method, are iterated until the final results of the biased dynamics is satisfactory. 
This can be slow as the enhanced sampling needs to be rerun for each iteration of the CV.
Once equipped with a low-dimensional representation of the chemical reaction, the enhanced sampling algorithms usually introduce a biasing potential, which is a function of the identified CVs and modifies the original PES by lowering the reaction barrier. 
If CVs and enhanced sampling technique are chosen well, the biased simulation will significantly increase the occurrence of reactive events, and subsequent analysis will allow us to understand the reaction mechanism and to calculate reaction barrier and rate. 

Simulations that are fully differentiable have been developed for optimization, control, and learning of motion,\cite{Degrave2016ARobotics, deAvilaBelbute-Peres2018End-to-EndControl, Hu2019Taichi, Hu2020DiffTaichi:Simulation} but also for the learning and optimization of quantities of interest in molecular dynamics \cite{Wang2020DifferentiableLearning, Ingraham2019LearningSimulator, Greener2021DifferentiableProteins}. 
Differentiating through simulations comes naturally from the optimization of path-dependent quantities (the famous Brachistochrone curve problem is included in \cref{ape:brachisto} for the novice reader). 
If the minimization of a loss function cannot be formulated separately for every point in the path, then optimization has to include the whole path leading up to it. 
While the results of DiffSims are often promising, it is well known \cite{Metz2021GradientsNeed} that naïvely backpropagated gradients may vanish or explode, and thus not lead to a useful parameter update. 
How to control their behaviour remains an open challenge. 
This problem of differentiable simulations is associated with the spectrum of the system's Jacobian\cite{Metz2021GradientsNeed, Galimberti2021HamiltonianDesign} and closely connected to the chaotic nature of the simulated equations. 
Therefore, in order to employ path differentiation, one needs to find ways to produce well behaving and controllable gradients.
In this contribution, the loss gradient behaviour is thoroughly investigated, and a mechanism to control its fluctuations and magnitude is proposed.
Employing the improved DiffSims, we define a differentiable loss function that, when minimized, results in the robust training of a biasing potential, which enhances the sampling of reactive transitions without prior determination of CVs.   

The manuscript is structured as follows. In Section~2 we define molecular dynamics simulations biased with a learnable potential, introduce a formalism to describe chemical reactions using path integrals, and outline the concept of differentiable simulations.
We discuss the current challenges and limitations of DiffSims in Section~3 and propose novel techniques to resolve them.
In section~4 we outline the practical implementation of our method.
In Section~5, we demonstrate the usefulness of DiffSims in the context of chemical reactions by training the bias function promoting barrier crossing for the well-studied Muller-Brown potential as well as the alanine-dipeptide molecule. 

\section{Problem Definition}

Molecular dynamics is commonly used to explore reaction processes on a atomistic level. 
Let the column vector $\xx \in \mathbb{R}^{N}$ denote the mass-weighted coordinates of the system and $\pp$ the conjugate momenta.
The particle motion is simulated using Hamiltonian equations with potential energy function $U_0(\xx)$. 
\begin{equation} \label{eq:md_standard}
\begin{split}
   \dot{\xx}(t) &= {\pp(t)} \\
   \dot{\pp}(t) &= -\frac{\partial U_0(\xx(t))}{\partial \xx}
\end{split}
\end{equation}
These equations conserve energy and are purely reversible with respect to time. 
However, it is common in molecular modeling not to work with the micro-canonical ensemble but rather with the canonical ensemble that conserves temperature \cite{Callen1998ThermodynamicsEd}. 
This is realized by using a thermostat coupled to the system. 
In this work, we choose the Langevin thermostat because of its implementational simplicity and its favorable properties with respect to differentiating along the computational graph, as will be shown later (see~\cref{sec:PatialBackprop}). 
In Langevin dynamics, the thermostat is coupled to the system through the friction constant $\gamma$ \eqref{eq:biased_dyn}.

To increase the probability of the barrier crossing we modify the PES with a learnable bias term $B(\xx, \theta)$
\begin{equation} \label{eq:biased_potential}
    U(\xx,\theta) = U_0(\xx) + B(\xx, \theta),
\end{equation}
where the biasing function is parameterized by $\theta$, which we aim to train to increase the frequency of reaction events.
The biased dynamics evolve according to
\begin{equation} \label{eq:biased_dyn}
\begin{split}
   \dot{\xx}(t) &= {\pp(t)} \\
   \dot{\pp}(t) &= -\frac{\partial U(\xx(t))}{\partial \xx} - \gamma \pp(t) + \sqrt{2 \gamma k_B T} \RR(t),
\end{split}
\end{equation}
where $k_b$ is the Boltzmann constant, $T$ the absolute temperature of the bath, and $\RR(t)$ a Gaussian process.

\subsection{Formal characterization of (chemical) reactions}

It is often suitable to use general curvilinear coordinates and not simply Cartesian or mass-weighted coordinates to describe reactions. 
Common are internal coordinates such as interatomic distances, angles, or dihedrals, as they are invariant with respect to system rotation and translation. 
These special coordinates are denoted with $\xxi(\xx) \in \mathbb{R}^M$ and $M \leq N$.

The wells $W_\alpha$ of reactant (-1) and product (1), divided by a reaction barrier, are characterized by the set of points $\Gamma_\alpha$ ($\alpha=-1,1$), which correspond to the equilibrium configurations of reactants and products, i.e., we expect an unbiased simulation on the PES $U_0(\xx)$, to stay in these wells with a very high probability. 
We approximate the wells with a multivariate normal distribution. 
From short, unbiased simulations, we estimate mean $\mean_\alpha$ and covariance matrix $\varm_\alpha $. We consider a point to be part of a well if the probability of the point belonging to the distribution is above some chosen probability threshold. 
\begin{equation}
    W_\alpha = \left\{\xx\ |\ (\xxi(\xx) - \mean_\alpha)^T \Sigma_\alpha^{-1} (\xxi(\xx) - \mean_\alpha)
    < \epsilon \right\} \ ,
    \label{eq:ball}
\end{equation}
where epsilon can be obtained from $\chi^2$ distribution. 
The indicator function for a well is
\begin{equation}
  \mathbb{1}_\alpha(\xx) =
  \left\{
    \begin{array}{l c r}
        1 & \text{for} & \xx \in W_\alpha   \\
        0 & \text{for} & \xx \notin W_\alpha .
    \end{array}
  \right.
  \label{eq:indicator}
\end{equation}
%
%
In this manuscript, we only consider transitions between two wells, $W_{-1}$ and $W_1$. 
Additional basins would be handled analogously. 
The (escape) probability $p_\alpha$, within a specified time interval $(t_0, t_e)$ of a transition $W_{-\alpha} \rightarrow W_\alpha$ is defined as
\begin{equation}
    p_\alpha = P\left(\int_{t_0}^{t_e} \mathbb{1}_{\alpha}(\xx(t)) \, \mathrm{d}t > 0 \; \middle|\; \xx(t_0) \in W_{-\alpha}\right) \ ,
\label{eq:init_pob}
\end{equation}
where $t_0$ is the start and $t_e$ the end time of the trajectory $\XX$.
This can be understood as the probability of finding at least one point in $W_\alpha$ of a trajectory that has started in $W_{-\alpha}$. 
Our objective is to increase both $p_1$ and $p_{-1}$ simultaneously to a level where both events can be observed frequently on a typical simulation time scale. 

\subsection{Optimizing the probability}
The form of the probability in \eqref{eq:init_pob} is not usable for differentiable optimization and needs to be recast to a differentiable, continuous form.
Under suitable regularity conditions, we can replace the expression of \eqref{eq:init_pob} with 
\begin{equation}
    p_\alpha = P \left(\sup_{t<t_e}( \mathbb{1}_{\alpha}(\xx(t)) > 0\ \middle|\; \xx(t_0) \in W_{-\alpha} \right).
\end{equation}
%
%
%
We can then define a soft loss function that is continuous everywhere and differentiable for any trajectory $\XX$ with $\xx(t_0) \in W_{-\alpha}$ as
\begin{equation}
    L = L_{\xxi_\alpha} = 
    \begin{cases}
        0  & \text{if}\ \exists \ \xx(t) \in W_\alpha \\
        \min\limits_{t_0<t<t_e} (\xxi(\xx(t)) - \xxi_{\alpha})^2 & \text{otherwise}
    \end{cases}
    \ ,
\label{eq:actual_loss}
\end{equation}
where for each trajectory a random single $\xxi_\alpha \in \Gamma_\alpha$ is selected for the loss function by running a short, unbiased simulation. 
Choosing random targets is done to increase the configuration space the simulation is forced to cover, avoiding targeting a particular point, thus making the optimization more robust.
Minimizing this loss function leads to a maximization of the probability \eqref{eq:init_pob} and can be seen as the minimization of a path-dependent integral. 
In the following Section, we will define a method that can be employed to minimize \eqref{eq:actual_loss}. 
Note that the loss function is defined only for one point of the trajectory and is influenced by the dynamics of every point that proceeds it. 
\subsection{Differentiable simulations}
For the problem at hand, the parameters $\theta$ of the bias potential \eqref{eq:biased_potential} have to be optimized such that the loss \eqref{eq:actual_loss} is minimal.
For a differentiable simulation, the information that we gain by differentiating the loss function at the point where it is defined can be used for optimization along the whole trajectory. 
While we could proceed by considering the simulation as a forward process, saving the computational graph for the entire path would be extremely memory-demanding. 
Instead, the optimization process can be conveniently reformulated using the adjoint equation and resulting adjoint vectors, using which the system dynamics can be run backwards to an arbitrary time, leading to memory-savings and the ability to adjust extent of backpropagation based on the sought-for dynamical scale (see Section~3). 
We employ the framework and notation adapted recently for neural networks\cite{Chen2018NeuralEquations} from the original work by \cite{LevSemenovichPontryagin1962TheProcesses}.

We propagate the state $\zz(t)=(\xx(t), \pp(t))$ using the biased Langevin dynamics  where the right side of \eqref{eq:biased_dyn} shall be denoted as $f(\zz(t), \theta) = \dot{\zz}(t)$. 
Propagating $\zz(t)$ using $f(\zz(t))$ is called the \textit{forward process}. 
Notice that $f(\zz(t))$ has no explicit time dependence (only through $\zz(t)$). 
We define the adjoint vectors for this equations as 
\begin{equation}
    \bva(t) = \frac{\partial L}{\partial \zz(t)} \ .
    \label{eq:adjoint}
\end{equation}
To solve \eqref{eq:adjoint} we introduce the new time $\tau \in (0, \tau_{e})$ such that $\zz(\tau=0) = \zz(t=t_e)$ and $\zz(\tau=\tau_{e}) = \zz(t=0)$. 
This backward flowing time reflects that the loss is not influenced by any points further in forward time. 

In forward moving time $t$, the adjoint vectors obey the equations
\begin{equation}
\begin{split}    
    \bva(t_e) &= \frac{\partial L}{\partial \zz(t_e)} \\
    \dot{\bva}(t) &= - \bva(t)^T \frac{\partial f(\zz(t), \theta)}{\partial \zz},
\end{split}
\end{equation}
or in backward going time $\tau$, the equation
\begin{equation} \label{eq:adjoint_tau}
\begin{split}    
    \bva(\tau = 0) &= \frac{\partial L}{\partial \zz(\tau = 0)} \\
    \dot{\bva}(\tau) &= \bva(\tau)^T \frac{\partial f(\zz(\tau), \theta)}{\partial \zz}.
\end{split}
\end{equation}
The total gradient of the loss function with respect to bias parameters is then obtained by
\begin{equation}
    \frac{\partial L}{\partial \theta} = \int_{0}^{\tau_e} \bva(\tau)^T \frac{\partial f(\zz(\tau), \theta)}{\partial \theta} \mathrm{d}\tau
    \ .
\label{eq:backprop_jvp}
\end{equation}
While solving \eqref{eq:adjoint_tau}, $\zz(\tau)$ can be either saved or reconstructed by running dynamics \eqref{eq:biased_dyn} backward, depending on the memory and computational trade-off we would like to maintain. The algorithm for running the adjoint method for the dynamics that includes random noise is developed and analyzed in \cite{Li2020ScalableEquations}.
\section{Challenges and Solutions}
\subsection{Challenges}
Ideally, one would simulate the biased dynamics \eqref{eq:biased_dyn}, compute the loss \eqref{eq:actual_loss}, backpropagate by solving \eqref{eq:adjoint_tau}, and after a number of training epochs obtain the biasing potential that enhances transitions. 
However, differentiable simulations at their current state cannot be used out of the box. 
There exist several issues that need to be addressed. 
\begin{enumerate}
    \item \textbf{Gradient control} \label{prob:bound} \\
        Significant effort has been devoted in the past years to understand the behavior of gradients that arise while optimizing neural network controlled differentiable simulations \cite{Suh2022DoGradients, Huang2021PlasticineLab:Physics, Metz2021GradientsNeed}. 
        Some of the main challenges in this respect are the explosion or the vanishing of gradients when training deep neural networks. 
        A differentiable simulation can be arbitrarily deep, however, it can be challenging to backpropagate complex Hamiltonians in a controllable manner to such depths.
        In fact, it is possible to construct a simple Hamiltonian that gives rise to exploding gradients when using \eqref{eq:biased_dyn} \cite{Galimberti2021HamiltonianDesign}. 

    \item \textbf{Multiscale Problem} \label{prob:multi}\\
        The dynamics on $U_0(\xx)$ can include very high and very low frequency motions. However, only the slow dynamics should be controlled by the trainable bias, as those are associated with the sought-after chemical reactions. 
        High frequency modes, e.g., hydrogen vibrations in the case of molecules, usually do not contribute to the reaction mechanism. 
        Avoiding fitting such fast fluctuations is desirable as it reduces the noise in the gradients used for DiffSim training. 
        
    \item \textbf{Chaotic behaviour} \label{prob:chaos} \\
        One important property of some Hamiltonian systems is the emergence of chaos \cite{PercivalI1987ChaosSystems}. 
        Small changes in initial conditions result in exponentially different trajectories.
        Great care must be taken to predict and control the behaviour of such systems. 

    \item \textbf{One large parameter update per trajectory} \label{prob:minib} \\
        Differentiable simulations in their original formulation produce one update per trajectory. Obtaining a sufficient number of gradient updates can be very expensive when long trajectories are required.

\end{enumerate}

All these challenges are addressed by the present work.
We show how to efficiently learn slow dynamics necessary for the investigation of chemical reactions while keeping gradients under control.

\subsection{Partial backpropagation}
\label{sec:PatialBackprop}

To reduce the complexity and level of detail in the equations, the computational graph is pruned such that backpropagation occurs only in the momenta. 
This is realized by adoption of the \textit{.detach()} operator introduced, e.g., in Refs.~\cite{Foerster2018DiCE:Estimator, Schulman2015GradientGraphs, Zhang2019AutomaticApplication}, which stops the flow of the gradient through $\xx$. 
The backpropagation only in momenta can be reasoned as follows:
\begin{itemize}
\item For timescales $\Delta \tau$, typical for the slow dynamics in the system, we assume $\xx$ to linearly approach the target value. 
This discards fast oscillations in $\xx$. In other words, for timescales $\Delta \tau$ we expect the change in $\xx$ to be linear in time
\begin{equation}
    \frac{\partial \xx(\tau)}{\partial \pp(\tau)} = \Delta \tau \ ,
\end{equation}
which corresponds to the equation for $\xx$ in the form
\begin{equation}
    \xx(\tau) = \xx(\tau - \Delta \tau).detach() + \pp(\tau) \Delta \tau \ .
    \label{eq:detached_ev}
\end{equation}
\item The use of the detach operator reduces the backpropagated ODE to first order in time, neglecting all higher order terms. 
Hence, the backpropagation follows a diffusion type equation, and any high frequency oscillations are removed. 
The change in adjoint dynamics can be seen in the Figure~\ref{fig:adjoint_comp}.
\item The modified dynamics are well behaved with respect to the magnitude of the gradients as shown in the theorems.
\end{itemize}

\begin{figure}
    \centering
    \includegraphics[width=\linewidth]{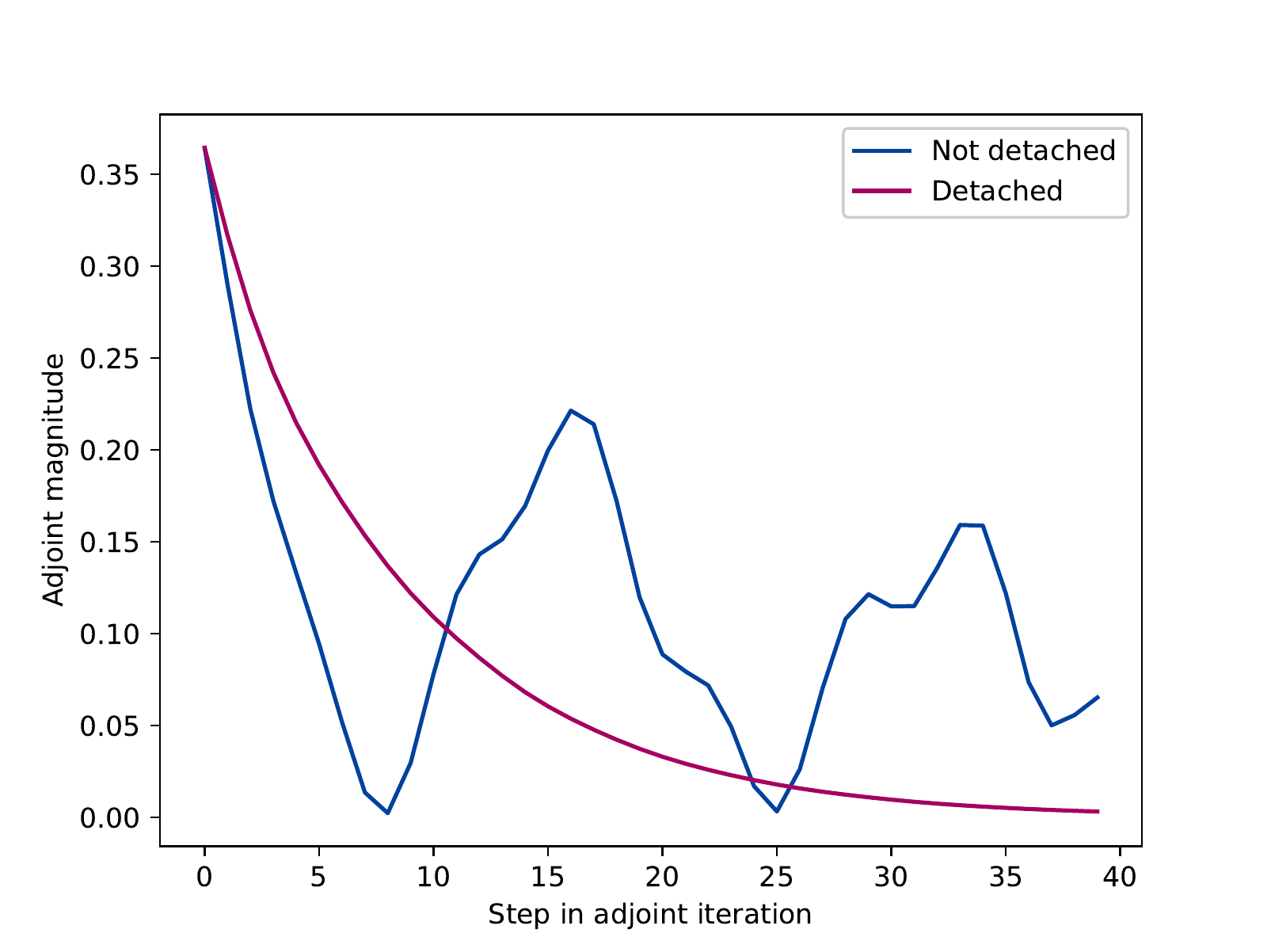}
    \caption{Comparison of the adjoint evolution for original and partially detached graphs for simulations on the 2D Müller-Brown potential with parameters given in \cref{ape:diffsimparams}.}
    \label{fig:adjoint_comp}
\end{figure}

The introduction of the \textit{.detach()} operator reduces the number of equations for which the adjoint is calculated and through which the loss is backpropagated.
Position $\xx(\tau)$ is no longer an independent variable only a function of $\pp(\tau)$.
The detached $\xx$ from the previous timestep is treated as a constant in \cref{eq:detached_ev}. 
The adjoint dynamics is then calculated in only one variable $\pp(\tau)$ and only the evolutionary equation for this variable is considered. 
To simplify the discussion of adjoints we will not discretize the backward equation, but keep it in the continuous form. 
The use of the \textit{.detach()} operator simplifies the adjoint time derivative to
\begin{equation} \label{eq:reduced_adjoint}
\begin{split}
    \dot{\bva}(\tau) 
    &= \bva^T \frac{\partial}{\partial \pp(\tau)} \left( -\frac{\partial  U(\xx(\tau))}{\partial \xx} - \gamma \pp(\tau) + \sqrt{2 \gamma k_B T} \RR(\tau) \right) \\ 
    &= - \bva^T(\tau) \left( \frac{\partial^2 U(\xx(\tau))}{\partial^2 \xx} \frac{\partial \xx(\tau)}{\partial \pp(\tau)}  - \gamma \mathbf{I}\right)  \\
    &= - \bva^T(\tau) \frac{\partial^2 U(\xx(\tau))}{\partial^2 \xx} \Delta \tau - \gamma \bva(\tau) \ .
\end{split}
\end{equation}
Let us now formulate the property of the adjoints that will be useful when designing numerical methods and also gives us some assurance of the non-diverging gradient dynamics. 
Consider a trajectory $x_t$ generated by \eqref{eq:biased_dyn} with time $t$ in a possibly infinite time interval $I \subset (-\infty, \infty)$. 
The loss \eqref{eq:actual_loss} is defined for a point $\xx_{t_L}$. To optimize $B(\xx, \theta)$, we need to backpropagate the gradient of this loss through every point proceeding $\xx_{t_L}$, using backwards flowing time $\tau$. 
To summarize the notation and to set the stage for the proof of finite gradient update, we introduce the following definitions:
\begin{definition} [Differentiable Trajectory]
A Differentiable Trajectory $\mathcal{T}$ is defined by the following quadruple $(\zz(t), L(\zz(t_L), f(\zz(t)), \tilde{f}(\zz(t)))$: Let $\zz(t) \in \Omega_{\xx} \times \Omega_{\pp}$ where $\Omega_{\xx} \subset \mathbb{R}^N$ and $\Omega_{\pp} \subset \mathbb{R}^N$ for $t \in (t_i, t_e)$, where $-\infty \leq t_i < t_e \leq \infty$ be the sequence of states generate by the dynamics $f(\zz(t))$ from a certain initial state $\zz(t_0), t_0 \in [t_i, t_e]$. 
We define a loss function $L(\zz(t_L))$ in time $t_L$. The gradient dynamics of the loss function is guided by the dynamics $\tilde{f}(\zz(t))$ that includes possible $.detach()$ operators. The backward dynamics is represented in the reverse flowing time $\tau$ starting from $\zz(t_L) = \zz(\tau=0)$ to $\zz(\tau_e) = \zz(t_i)$. 
\end{definition}
And we define a Diffusive Differentiable Trajectory by
\begin{definition} [Diffusive Differentiable Trajectory]
A Differentiable Trajectory $\mathcal{T}$ constructed by dynamics \eqref{eq:biased_dyn} equipped with a backward dynamics \eqref{eq:reduced_adjoint} and a loss function \eqref{eq:actual_loss} is called a Diffusive Differentiable Trajectory, denoted by $\mathcal{T}_d$. 
\end{definition}

\begin{property} [Finite gradient update] \label{prop:convergence}
Let $\gamma$ be sufficiently high. Let 
\begin{equation} \label{eq:assumptions_finite}
U(\xx) \in C^2(\Omega_{\xx}) \ \operatorname{and} \ \frac{\partial f(\zz(\tau),\theta)}{\partial \theta} \ \operatorname{bounded}
\end{equation}
Then the gradient update for a loss function in a Diffusive Differentiable Trajectory is finite for every (possibly infinite) $\tau_e$. 
\end{property}
The essence of the proof and specification of the sufficient conditions for $\gamma$ are addressed in the following theorem. 
\begin{theorem}  [Converging adjoints] \label{thm:adjoint_convergence}
$\mathcal{T}_d$ be a Diffusive Differentiable Trajectory. Let $U(\xx) \in C^2(\Omega_{\xx})$ and denote the spectrum of its hessian $\frac{\partial^2 U(\xx(t))}{\partial^2 \xx}$ by $\lambda_i(\xx(t))$. Define $\lambda_{min}$ as
\begin{equation}
\lambda_{min} = \inf_{\tau>0} \min_i \lambda_i(\xx(\tau)).
\end{equation}
Then for every $\gamma$ that fulfills: $\left( {\Delta \tau}\lambda_{min} + \gamma \right) = \epsilon > 0$, it holds:
\begin{equation} \label{eq:upper_bound}
    \forall \tau > 0: \norma{\bva(\tau)}^2 \leq \norma{\bva(0)}^2 e^{-2 \epsilon \tau}
\end{equation}
\end{theorem}
We proof the theorem in the \cref{ape:adjoint_proof}. There is a useful corollary of the above
\begin{corollary}
Under the assumptions of the theorem \ref{thm:adjoint_convergence}, $\norma{\bva(\tau)} \in L^r(0, \tau_e), \ r \in [1, \infty]$.
\end{corollary}
\begin{proof}
Case $r = \infty$ is trivial as the square root of the upper bound \eqref{eq:upper_bound} is still finite $\forall \tau$. Let us now consider only $r \in [1, \infty)$. 
\begin{align}
    \norma{\bva(\tau)}_{L^r(0,\tau_e)}^r & = \int_0^{\tau_e} \norma{\bva(\tau)}^r \leq \norma{\bva(0)}^r \int_0^{\tau_e} e^{-\epsilon \, r \, \tau} \nonumber \\
    & = - \frac{\norma{\bva(0)}^r}{\epsilon \, r} \left[ e^{-\epsilon \, r \, \tau} \right]_0^{\tau_e} 
\end{align}
Which is finite for every value of $\tau_e$ including $\infty$.
\end{proof}

\begin{proof} [Proof of the finite gradient update \ref{prop:convergence}]
Since we know that $\bva(\tau) \in L^1(0,\tau_e)$ from the previous corollary and that $\frac{\partial f(\zz(\tau),\theta)}{\partial \theta}$ bounded from the assumption, it is now trivial to show 
\begin{align}
    \frac{\partial L(\zz_{t_L})}{\partial \theta} & = \int_0^{\tau_e} \bva(\tau) \frac{f(\zz(\tau),\theta)}{\partial \theta} d \tau 
    \nonumber \\
    & \leq 
    \sup_{\zz(\tau) \in T} \norma{\frac{\partial f(\zz(\tau),\theta)}{\partial \theta}} \int_0^{\tau_e} \bva(\tau) d\tau,
\end{align}
which is finite. 
\end{proof}
This property allows us to backpropagate the dynamics without exploding gradients as long as $\gamma$ is chosen large enough. 
The exponential scaling of the adjoints also indicates that once we identify the point where the loss function will be calculated, we only need to consider a handful of points before $\bva(\tau)$ essentially vanishes. Any further adjoint propagation does not significantly contribute to the gradient update. 
This is intuitively desirable, as for the noisy equation \eqref{eq:biased_dyn} the loss function information becomes diluted as we backpropagate. 
Keeping only recent data points thus introduces a natural cutoff to the information we use for optimization. 

The theorem also gives more insight into when such backpropagation may lead to exploding gradients. 
If the expression $\left( \Delta \tau\lambda_{min} + \gamma \right) = \epsilon < 0$, then the upper bound may not hold, and gradients can increase exponentially. 
Strongly negative $\lambda_i$ of the hessian indicates a concave part in the potential landscape, which is generally problematic for control. 
However, with $\Delta \tau$ and $\gamma$, we have two robust dials to  ensure non-exploding adjoints.

In practice, we assume $\Delta \tau$ to be equal to the forward timestep. 
Investigating the impact of setting $\Delta \tau$ to multiples of the timestep is beyond the scope of this paper.

\subsection{Mini-batching the graph}
One of the problems associated with differentiable simulation is the low number of updates. 
Usually, only one gradient step is taken per trajectory, making the gradients averaged across the entire path and necessitating rather large learning rates to train the network in just a few updates. 
The problem can be alleviated by a technique we call \emph{graph mini-batching}. 
The idea is to calculate trajectory depended gradients first (the adjoints $\bva$) in one pass and then split them to mini-batches. 
The adjoints are then used as vectors in Jacobi-vector products~\eqref{eq:backprop_jvp} during backpropagation of the bias function evaluated in batches. 
The approach stabilizes learning and allows for much lower learning rates, better suited for training neural networks. 
An example of a use case is more thoroughly discussed in \cref{ape:graph_minibatch}.

\subsection{Summary} 
The use of the Langevin thermostat with reasonable $\gamma$ creates finite memory dynamics and therefore decaying adjoints. Employing also the \textit{.detach()} operator ensures that the adjoints vanish smoothly, without high frequency oscillations, thus making them bounded (solving \cref{prob:bound}) and ignoring fast motion, helping with \cref{prob:multi}.
Such a finite memory system is likely to be less chaotic, addressing \cref{prob:chaos}.
Splitting the loss gradient into random mini-batches obviously solves the point~4.

\section{Practical implementation}

It is important to promote the transition across the barrier equally. 
If only one direction is sampled, then one may end up with a "landslide" potential strongly tilted towards one minimum and not a diffusive behaviour. 
Therefore, we choose the following approach.
\begin{enumerate}
  \item Create a batch of $2l$ starting configurations, with $l$ in each well respectively.
  
  \item Run all trajectories simultaneously for a fixed number of time steps.
  
  \item Collect the loss \ref{eq:actual_loss} after all simulations have ended. 
  After $N$ initial steps, which serve as equilibration, we also calculate the minimal distance from the start to encourage the eventual return to the starting well and, thereby, true diffusive behavior. 
  Thus, we have two losses: A forward loss $L_f(\xx_{L_f})$ and start loss $L_s(\xx_{L_s})$ that are summed together with equal weights.
  
  \item Calculate adjoints and optimize the bias function $B(\xx, \theta)$ using the graph mini-batching technique. 
  Repeat from step 1 until convergence. 
\end{enumerate}
By running a large number of simulations concurrently, one can leverage the vectorization of the operations and reduce computational time.


\begin{figure*}[tb]
\begin{subfigure}[b]{0.3\linewidth}
    \centering
  \includegraphics[width=1.0\linewidth]{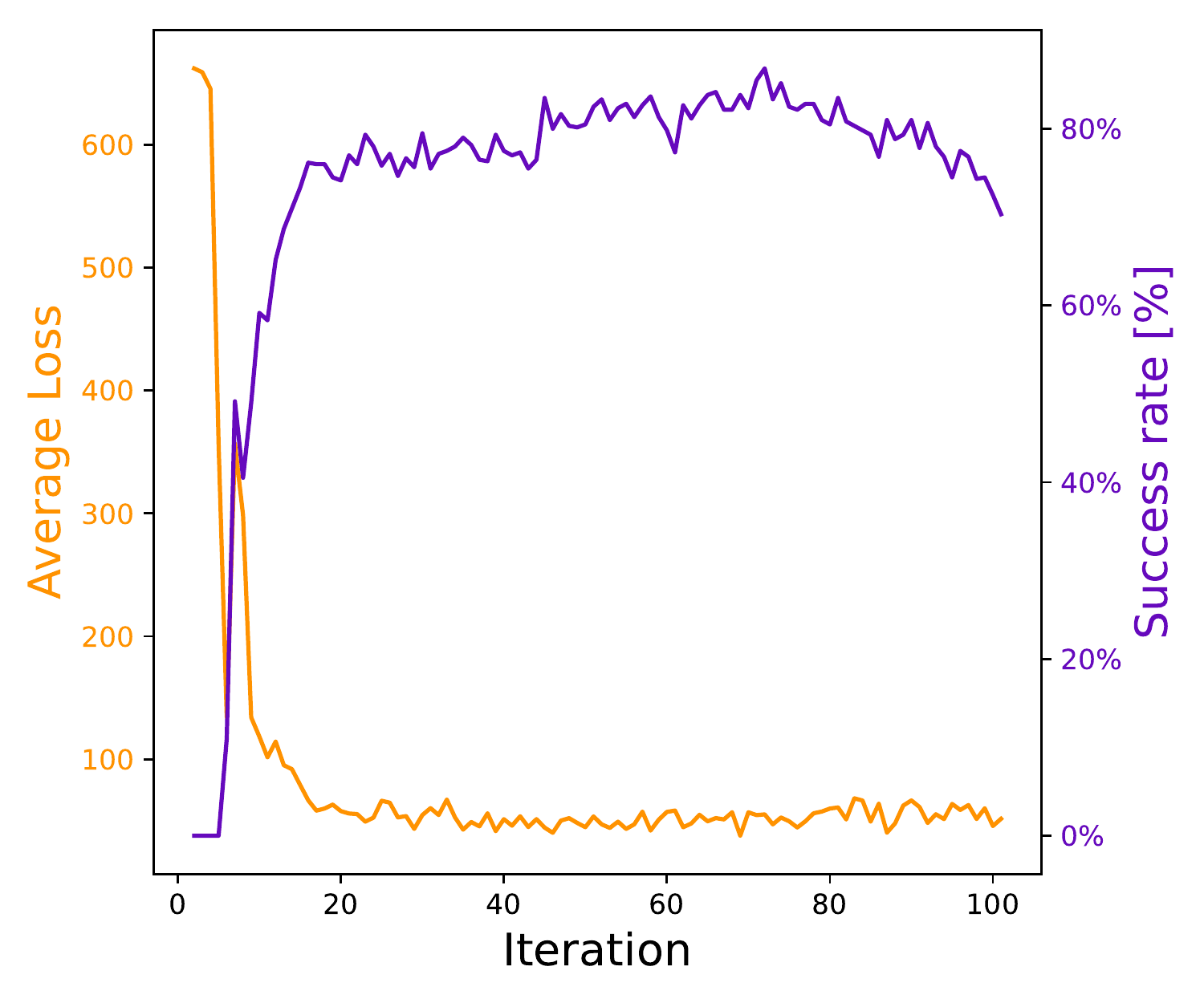}
  \phantomsubcaption
  \label{fig:2dsucrate}
\end{subfigure}
\centering
  \begin{subfigure}[b]{0.3\linewidth}
  \centering
  \includegraphics[width=1.0\linewidth]{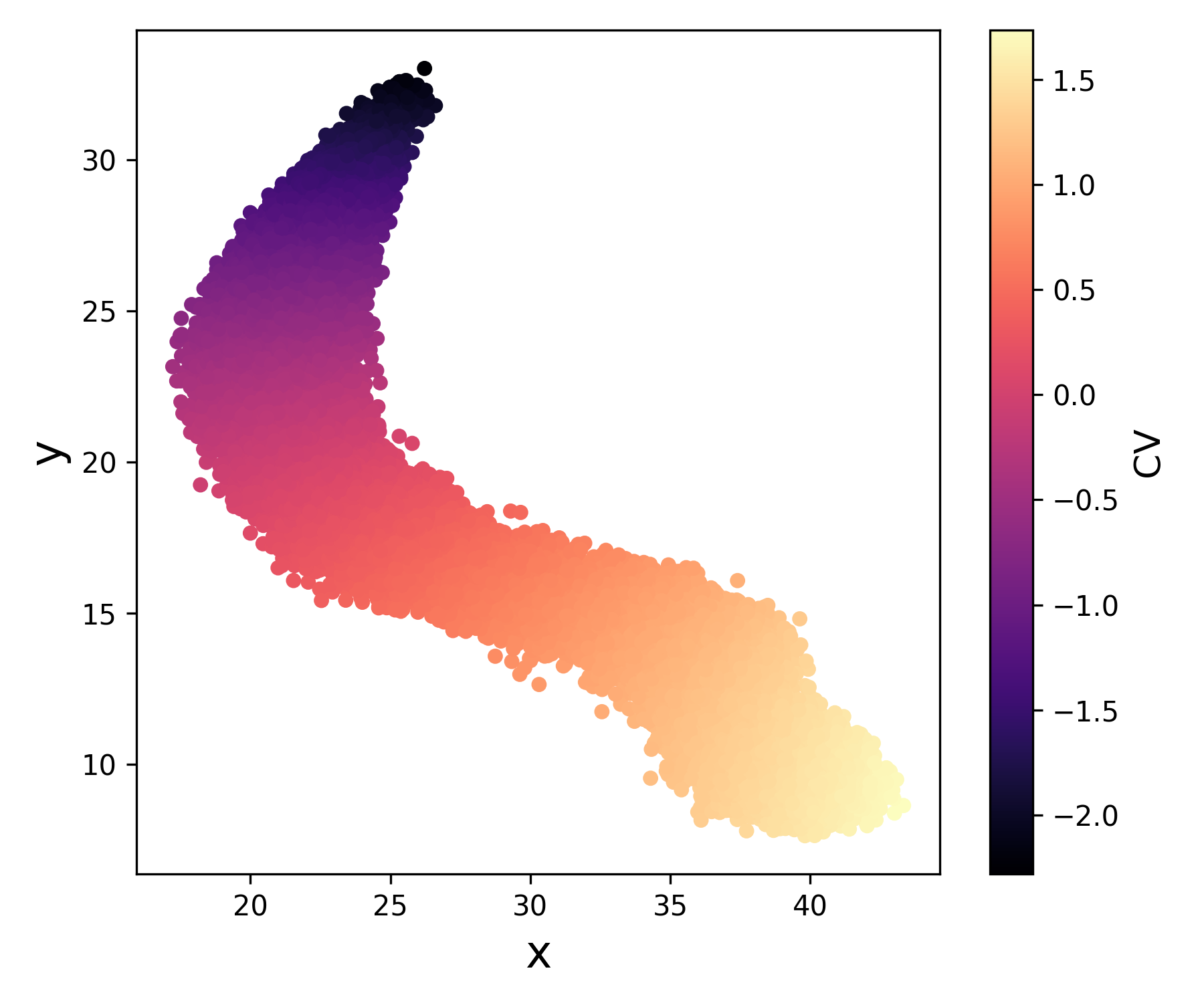}
  \phantomsubcaption
  \label{fig:2dcvs}
\end{subfigure}%
  \begin{subfigure}[b]{0.3\linewidth}
  \centering
  \includegraphics[width=1.0\linewidth]{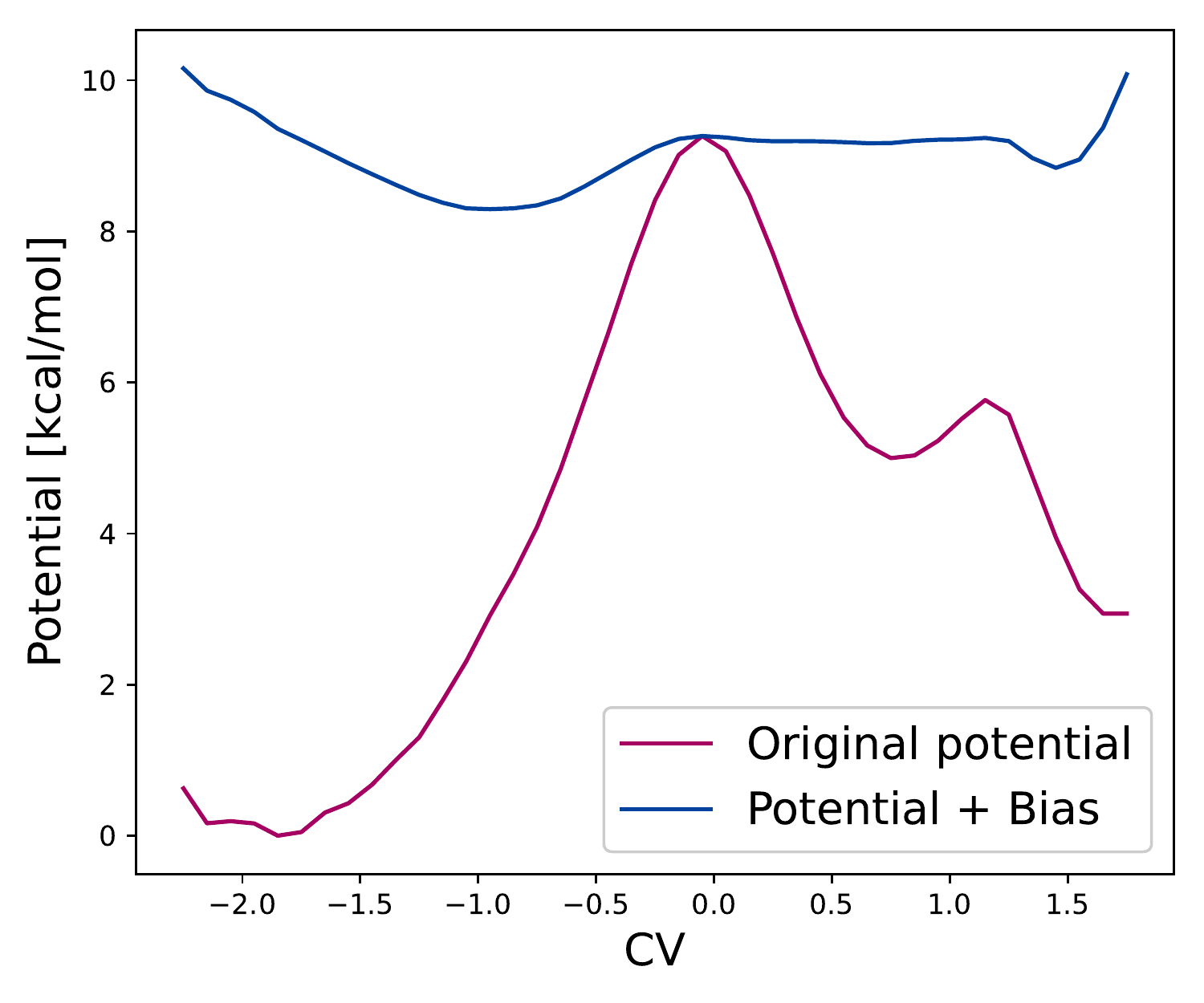}
  \phantomsubcaption
  \label{fig:2dpot}
\end{subfigure}
\vspace*{-5mm}
\caption{Postprocessing of the converged trajectory. \emph{left:} Loss functions and the probability of barrier crossing during the training progresses. \emph{middle:} Variational Autoencoder producing a collective variable by training on a fully diffusive trajectory. \emph{right:} Potential energy along the VAE collective variable with and without bias.}
\label{fig:2dresults}
\end{figure*}

\section{Results}

In this Section, we present the results of our novel DiffSim approach. 
First, we apply it to a commonly used two dimensional model PES, the Muller-Brown potential\cite{Muller1979LocationProcedure}, where any linear combination of the Cartesian coordinates does not yield a good CV. 
Then we lift this example to five dimensions by introducing three noisy DoFs demonstrating the efficiency of the approach in a higher-dimensional setup. 
Second, we investigate the benchmark system for enhanced sampling in molecular systems, alanine dipeptide (amino acid alanine capped at both ends). 
The two collective variables describing the metastable states are well known in the biophysics community, the backbone dihedrals $\phi$ and $\psi$. 
We will assume no such knowledge and generate the enhanced sampling simulation from all backbone dihedral angles as candidates in an end-to-end process.

\begin{figure}[bt]
\includegraphics[width=1.0\linewidth]{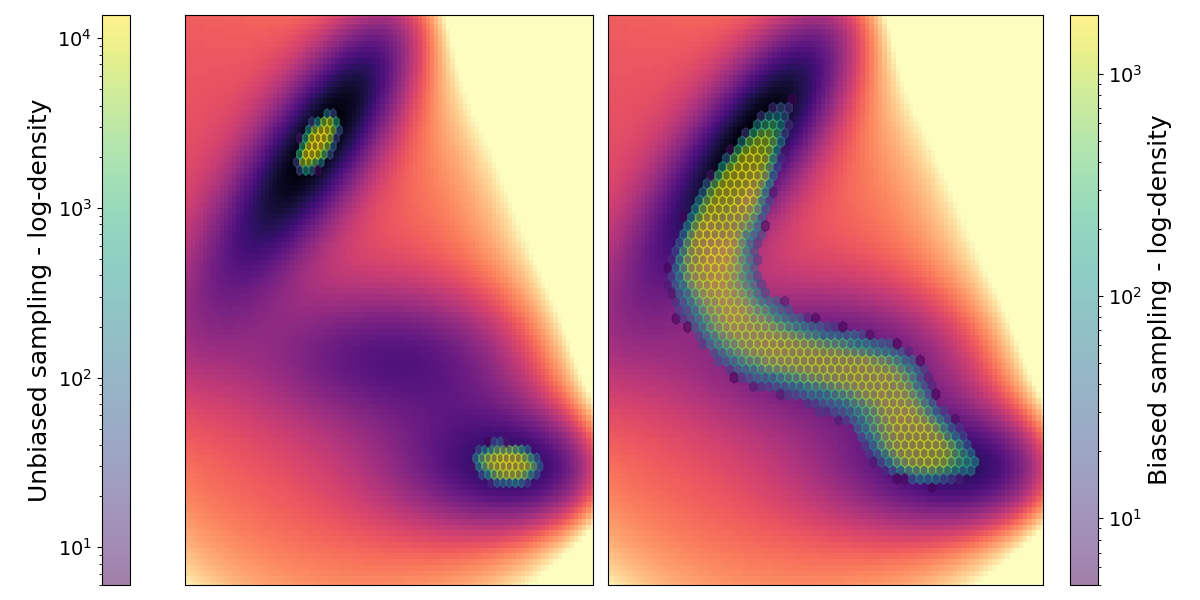}
\caption{Log-density of simulated points before (\emph{left}) and after the training  (\emph{right}) of bias function by differentiable simulations. 
The right plot shows how well all important regions are sampled after training. The background of the Figure is the $U_\mathrm{MB}(x,y)$, the underlying Muller-Brown potential.}
\label{fig:2ddensity}
\end{figure}

\subsection{2D Muller-Brown potential} \label{sec:2dexample}
The parameters of the commonly investigated 2D Muller-Brown PES \cite{Muller1979LocationProcedure, Sun2022MultitaskEvents} are given in the \cref{ape:mb}. 
For the bias potential, $B(\xx, \hh)$, we employ a grid of Gaussian functions, controlling their individual height. 
The biasing function is 
\begin{equation}
    B(\xx, \hh) = \sum_{i=1}^{n_g^2} h_i \exp \left( -\frac{\left(\xx - \xx_i^0\right)^2}{2 \sigma^2} \right)
\label{eq:gaussbias}
\end{equation}
with trainable $\hh$. 
Means $\xx_i^0$ are evenly distributed in the computational domain. 
With $n_g$ Gaussians along each dimension, only $n_g^2$ contributions to the total bias have to be calculated in two dimensions. 
After training the bias via DiffSim (parameters reported in \cref{ape:diffsimparams}), we obtain biased dynamics that generate increasingly many successful transitions between reactants and products along the transition path (see the evolution of the loss function and success rate during the training shown in \cref{fig:2dresults}). 
This leads to the log-density of the points along the transition path to even out significantly (see \cref{fig:2ddensity}).

We construct the CV by dimensionality reduction of frames from converged diffusive trajectories (well sampled transitions). 
To obtain a one dimensional CV describing the path, we use a Variational Autoencoder \cite{Kingma2013Auto-EncodingBayes} (architecture described in \cref{ape:diffsimparams}). 
The resulting CV is visualized in  Figure~\ref{fig:2dresults}. 
The CV distinguishes well and interpolates smoothly between products and reactants. 
Using this CV, the unbiased and biased PES are plotted as averages along the CV. 
It is easy to see in \cref{fig:2dresults} how effectively the PES has been flattened by the bias function.

\subsection{5D Generalization of Muller-Brown potential}
The situation is more complicated when additional harmonic degrees of freedom are included (see \cref{ape:general} for details). 
One may consider them to be, e.g., quickly oscillating hydrogen atoms that do not influence the reaction. 
As the ansatz of bias potential \eqref{eq:gaussbias} scales exponentially with the dimensionality of the problem, it cannot be used with the 5D version of the potential (\cref{ape:general}).
Instead, a fully connected neural network as a function of all five variables is employed, making training significantly harder. 
The results were postprocessed analogously as the two dimensional case, see Figure \ref{fig:5dresults}. 
Finally, the biased system converges to a success rate of 65~\%. 
As before, we observe that the potential was relatively flattened and transitions occur with high probability. 
The results are not as good as in the 2D case, both due to the built-in noisiness of the dynamics (\cref{ape:general}) and because of a rather crude approximation of the biasing potential with a simple fully connected neural network, which is harder to train than the Gaussian grid used in a 2D case. However, even under such circumstances, the DiffSim approach managed to discover the transformation path and sample it with reasonably high probability.

\begin{figure}
\centering
\begin{subfigure}[b]{0.9\linewidth}
\centering
\includegraphics[width=1.0\linewidth]{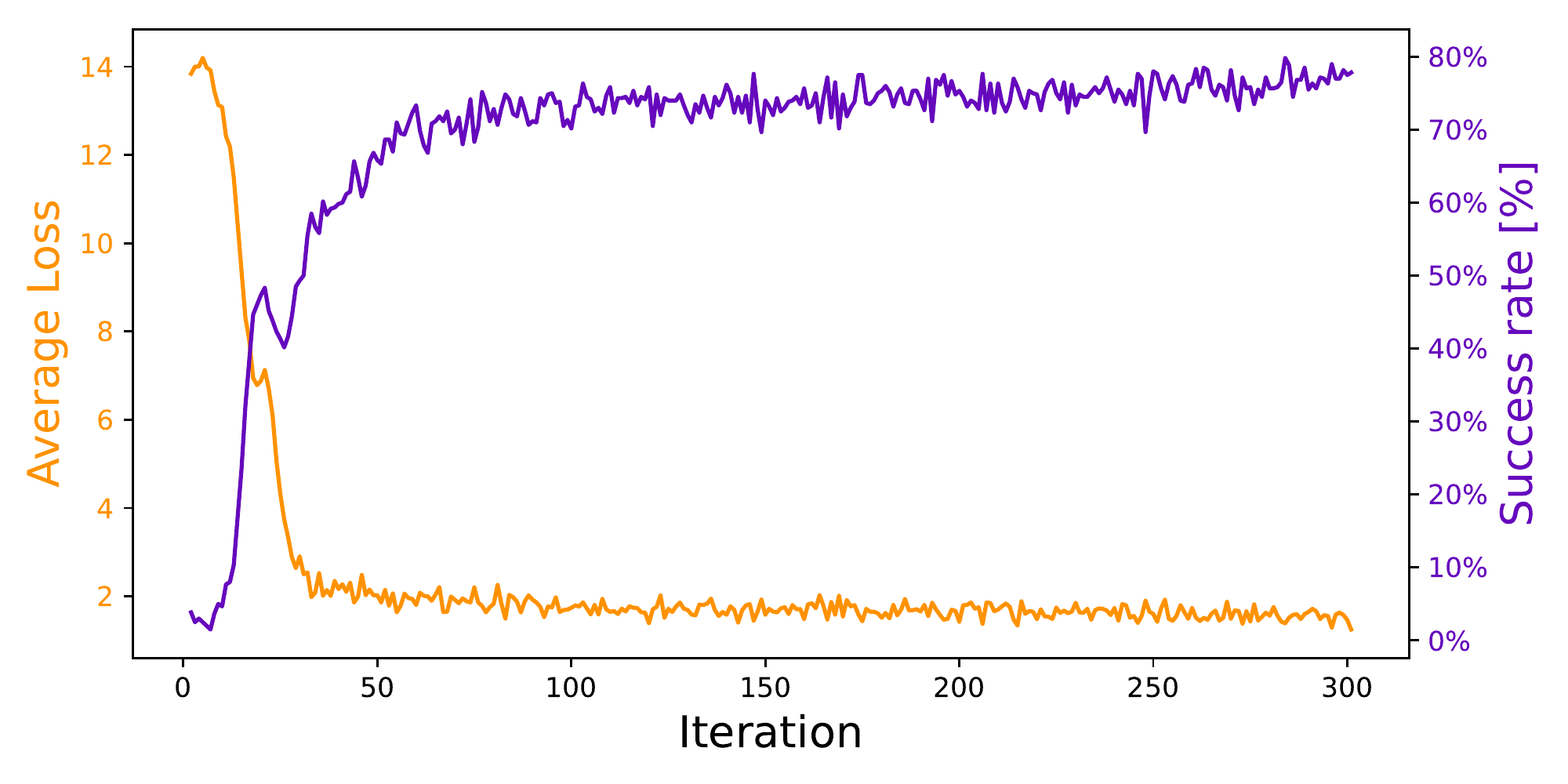}
\end{subfigure}
\begin{subfigure}[b]{0.9\linewidth}
\centering
\includegraphics[width=1.0\linewidth]{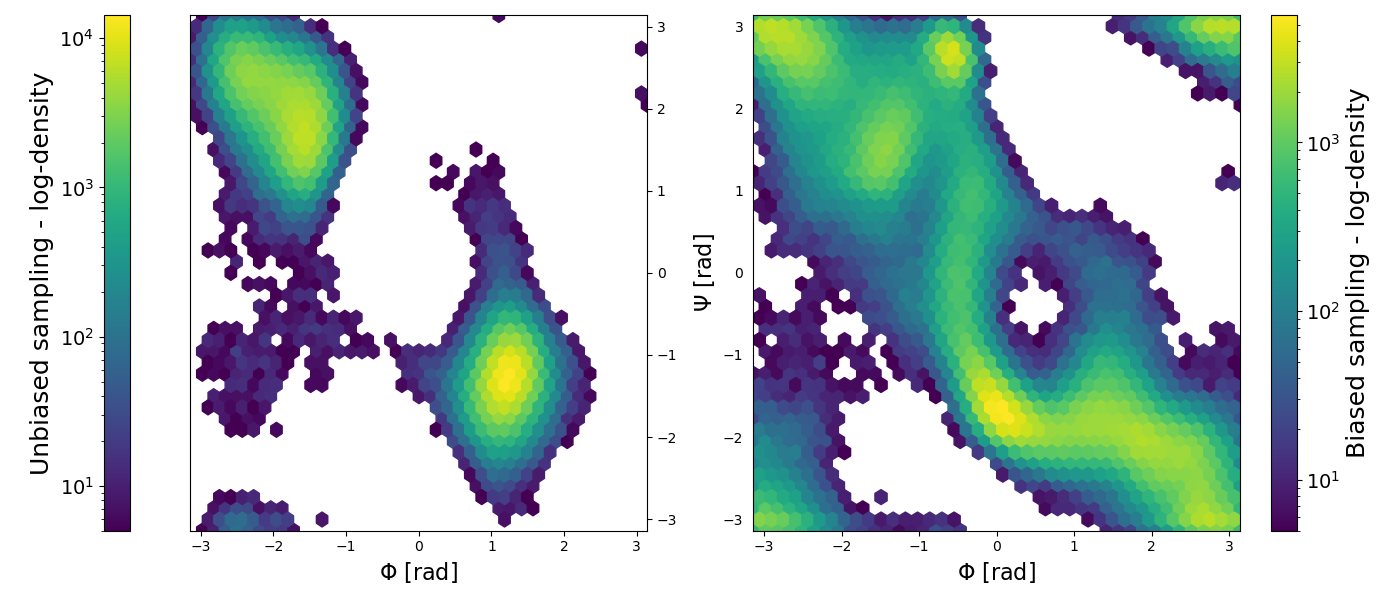}
\end{subfigure}
\caption{\emph{top row}: Metrics for the alanine dipeptide run. Total loss function and transition success rate. \emph{bottom row}: Log-density of simulated points before the training of bias by differentiable simulations (\emph{left}) and after the training converged (\emph{right}).}
\label{fig:ala2density}
\end{figure}

\subsection{Alanine dipeptide}
Alanine dipeptide is a simple model system exhibiting typical protein dihedral dynamics. 
Therefore, it has become an important benchmark to test and verify free energy calculation methods. The collective variables, the dihedral angles $\phi$ and $\psi$, are well known and the PES is rather complex with relatively low barriers \cite{Vymetal2010MetadynamicsStudy, Mironov2019ADipeptide}. We use this system to test the ability of our method to bias the dynamics along the important DoFs. The details about the location of minimas and biasing function inputs are reported in \cref{ape:diffsimparams}.

The progress of the training and log-density of points are reported in \cref{fig:ala2density}.
After training almost 80~\% of trajectories show a transition within 10~ps. 
By comparing the averaged bias potential with a potential of mean force (PMF) obtained using the reference metadynamics run (see \cref{ape:alaparams}) we can see that their shapes are similar (see \cref{fig:compareala}). 
This demonstrates that our differentiable simulations can unravel the transition paths and reaction barriers in the same way as a collective variable based method would, except without requiring prior knowledge of ideal CVs. 

\begin{figure}
\includegraphics[width=0.55\linewidth]{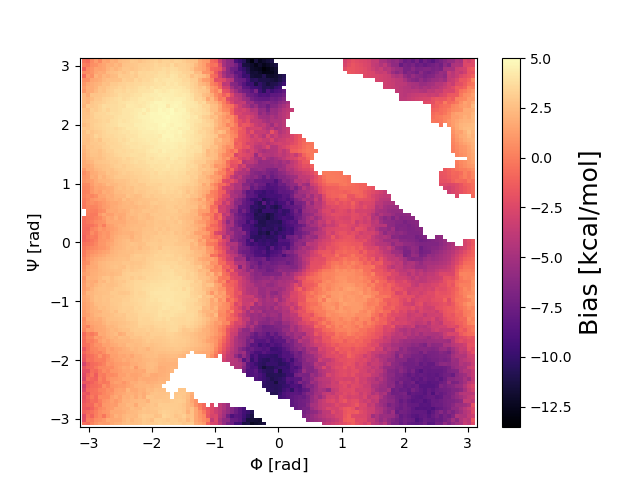}
\includegraphics[width=0.44\linewidth]{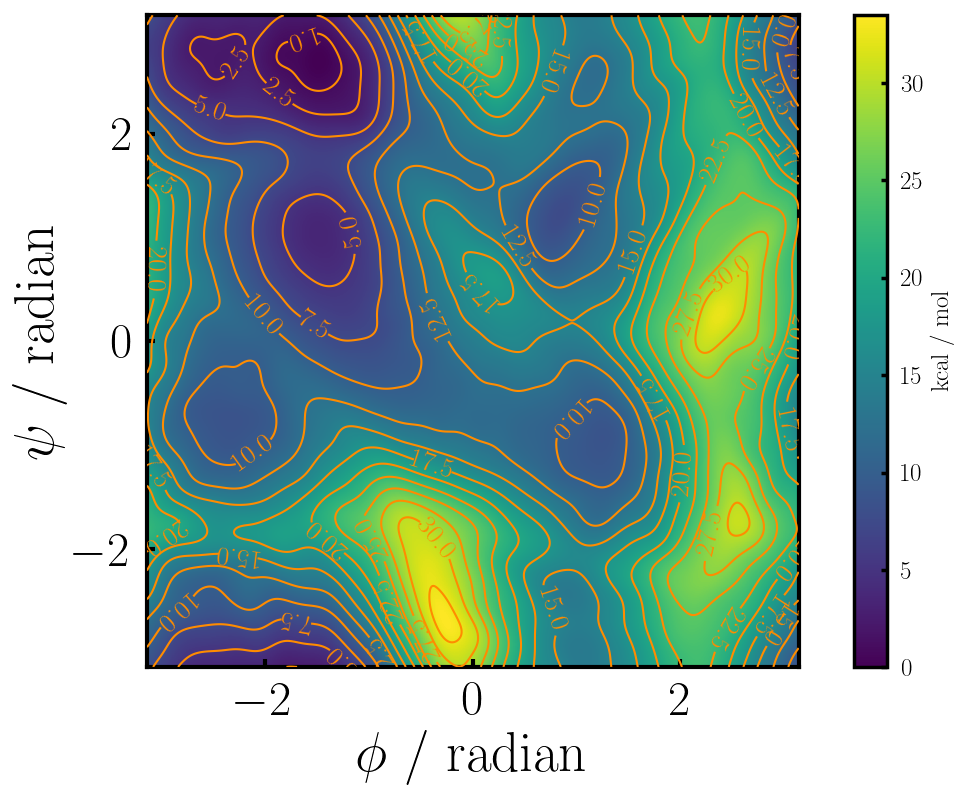}
\caption{\emph{left:} Average bias potential projected on the Ramachandran plane. 
White regions are without sufficient sampling to calculate bias potential.
\emph{right:} PMF of the $\phi$-$\psi$-plane.   }
\label{fig:compareala}
\end{figure}

\section{Conclusion}

This contribution presents advances in two distinct areas. 
First, it was described in detail how neural network controlled differential simulations (DiffSims) can be made robust and efficient. 
We have shown how the use of Langevin dynamics creates a finite memory horizon and therefore enforces decaying adjoints. The \textit{.detach()} operator pruned the computational graph and thereby ensured that the adjoints vanish smoothly, removing any high frequency oscillations.
The introduction of random mini-batches by breaking up the loss gradient made learning in the fashion of stochastic gradient descent possible, significantly stabilizing the training.

Second, the establishment of the robust and efficient neural network controlled DiffSims allowed us to successfully tackle an important open problem in computational chemistry - discovery and effective sampling of the rare event (chemical reaction) pathways.
Initially, a path integral loss was defined to measure the success of a molecular dynamics trajectory with regard to a crossing of an energy barrier, i.e., with regard to exhibiting a rare event.
This loss was then used to train a bias potential to discover and accelerate the chemical transitions without a need to guess a low-dimensional representation of the chemical reaction, i.e., the collective variable (CV), \textit{a priori}.
We showed the effectiveness of this approach by successfully biasing the dynamics on the Muller-Brown potential, which is a numerical benchmark for traditional enhanced sampling schemes and \emph{a priori} CV determination algorithms. Our method worked without any previous knowledge of the good CV, however, from the the biased trajectories exhibiting transitions a reduced representation, i.e., a CV,  can be constructed.
The quality of biasing and subsequent CV identification was practically perfect for the 2D case, and even the challenging 5D case with a significant amount of noise in the added harmonic DoFs converged, exhibiting a high probability of observing a transition event. 
Finally, a realistic chemical system (alanine dipeptide) was investigated.
The bias potential was constructed  considering dihedral angles as candidate degrees of freedom, including not only the two dihedrals commonly used as CVs but also other dihedrals that noised the transition.
Our method successfully generated biased trajectories, which exhibited sought-for transitions between the two target minima with high probability.

We have demonstrated that differentiable simulations with our innovations can handle not only model systems but also complex molecular motions. In the future, we intend to extend the tool to more challenging reactions with complicated transition paths, such as protein motion and chemical reactions with multiple intermediate steps or competing reaction paths. 

\section*{Acknowledgements}

M.S. was supported by project No. START/SCI/053 of Charles University Research program.
J.C.B.D. is thankful for the support of the Leopoldina Fellowship Program, German National Academy of Sciences Leopoldina, grant number LPDS 2021-08. 
L.G. acknowledges the support of Primus Research Program of the Charles University (PRIMUS/20/SCI/004).
R.G.-B. acknowledges support from the Jeffrey Cheah Career Development Chair.
We thank Michal Pavelka for discussions regarding the nature of multiscale problems.

\bibliography{diffsimrare}
\bibliographystyle{icml2022}

\newpage
\appendix
\onecolumn
\section{Proof of adjoint convergence theorem} \label{ape:adjoint_proof}
To prove Theorem \ref{thm:adjoint_convergence} we need to state one more lemma. 
\begin{lemma} \label{lmm:eigenlemma}
Let $\xx \in \mathbb{R}^N$, $U(\xx)$ scalar, real, $C^2(\mathbb{R}^N)$ function. Consider a hessian computed at $\xx_0$: $\frac{\partial^2 U(\xx_0)}{\partial \xx^2}$ with minimum and maximum eigenvalues $\lambda_{min}$ and $\lambda_{max}$ respectively. Then for any vector $\bvv \in \mathbb{R}^N$
\begin{equation}
\begin{split}
    \lambda_{min} \norma{\bvv}^2 \leq \bvv^T \cdot \frac{\partial U(\xx_0^2)}{\partial \xx^2} \bvv \leq \lambda_{max} \norma{\bvv}^2.
\end{split}
\end{equation}
\end{lemma}
\begin{proof}
We note that a hessian of a real continuous function is a symmetric matrix. Such a matrix is orthogonally diagonalizable and has real eigenvalues. The rest of the proof is a part of most standard linear algebra textbooks.
\end{proof}
We can now prove the Theorem \ref{thm:adjoint_convergence}.
\begin{proof}
We start by multiplying \eqref{eq:reduced_adjoint} by $2\bva$. This yields
\begin{equation}
    2\bva(\tau) \cdot \dot{\bva}(\tau) = {2 \mathrm{d} \tau} \bva^T(\tau) \cdot \frac{\partial^2 U(\xx(\tau))}{\partial^2 \xx} \bva(\tau) - 2 \gamma \norma{\bva(\tau)}^2
\end{equation}
and can be recast using $2 \bva^T(\tau) \cdot \dot{\bva}(\tau) = \dot{\overline{\norma{\bva(\tau)}^2}}$ (the norm is a standard vector 2-norm) to
\begin{equation}
     \dot{\overline{\norma{\bva(\tau)}^2}} = - {2 \mathrm{d} \tau} \bva^T(\tau) \cdot \frac{\partial^2 U(\xx(\tau))}{\partial^2 \xx} \bva(\tau) - 2 \gamma \norma{\bva(\tau)}^2
\end{equation}
Using Lemma \ref{lmm:eigenlemma} and, subsequently, the assumption of the theorem, we can estimate the upper bound of the time derivative as
\begin{equation}
    \dot{\overline{\norma{\bva(\tau)}^2}} \leq - 2  \left( {\mathrm{d} \tau} \lambda_{min} + \gamma \right) \norma{\bva(\tau)}^2 = -2 \epsilon \norma{\bva(\tau)}^2
\end{equation}
Using Gromwall lemma we can now estimate $\bva(\tau)$ easily as
\begin{equation} \label{eq1}
\begin{split}
   \norma{\bva(\tau)}^2 \leq \norma{\bva(0)}^2 \operatorname{exp} \left(- 2\int_0^\tau \epsilon\ \mathrm{d} t \right) = \norma{\bva(0)}^2 e^{-2\epsilon \tau}
\end{split}
\end{equation}
and since $\epsilon > 0$, the $\norma{\bva(\tau)}^2$ is bounded for all $\tau$.
\end{proof}
\section{Graph minibatching and adjoints} \label{ape:graph_minibatch}
To better explain the graph minibatching technique, let us consider a simple differential equation with trainable parameters $\theta$
\begin{equation}
    \dot{z} = f(\zz, \theta)
\end{equation}
Let us discretize the equation using a simple Forward Euler method such that it becomes
\begin{equation}
z_{n+1} = z_n + dt f(z_n, \theta).
\end{equation}
For simplicity consider a three step differentiable simulation $(z_0, z_1, z_2)$ such that 
\begin{align*}
z_2 = z_1 + dt f(z_1, \theta) \\
z_1 = z_0 + dt f(z_0, \theta)
\end{align*}
where a loss function is defined for the last point $L(z_2)$. Our goal is to find the gradient of $\frac{\partial L(z_2)}{\partial{\theta}}$. Let us derive
\begin{align*}
\frac{\partial L(z_2)}{\partial \theta} &= \frac{\partial L(z_2)}{\partial z_2} \frac{\partial z_2}{\partial \theta} \\
\frac{\partial z_2}{\partial \theta} &= \frac{\partial z_1}{\partial \theta} + dt \frac{ \partial f(z_1,\theta)}{\partial \theta} = \frac{\partial z_1}{\partial \theta} + dt \left( \frac{ \partial f(z_1,\theta)}{\partial z_1} \frac{\partial z_1}{\partial \theta} + \frac{\partial f(z_1,\theta)}{\partial \theta} \right)  \\
\frac{\partial z_1}{\partial \theta} &= \frac{\partial z_0}{\partial \theta} + dt \frac{ \partial f(z_0,\theta)}{\partial \theta} = dt \frac{ \partial f(z_0,\theta)}{\partial \theta}.
\end{align*}
Put together,
\begin{equation}
\frac{\partial L(z_2)}{\partial \theta} = \frac{\partial L(z_2)}{\partial z_2} \left[ dt \left(1 + dt \frac{ \partial f(z_1,\theta)}{\partial z_1} \right) \frac{ \partial f(z_0,\theta)}{\partial \theta} + dt \frac{\partial f(z_1,\theta)}{\partial \theta} \right].
\end{equation}
Meaning, when we optimize the biased function $f(z_n, \theta)$ We can split the derivative into two parts
\begin{equation}
\begin{split}
& \left[ \frac{\partial L(z_2)}{\partial z_2} dt \left(1+ dt \frac{ \partial f(z_1,\theta)}{\partial z_1} \right) \right] \frac{ \partial f(z_0,\theta)}{\partial \theta} \\
& \left[ \frac{\partial L(z_2)}{\partial z_2} dt \right] \frac{\partial f(z_1,\theta)}{\partial \theta}
\label{eq:jvpgrads}
\end{split}
\end{equation}
More steps can be obtained by continuing the iterations. One can easily see that the vectors we put into square brackets are actually the adjoints $a(z_n)$ from \eqref{eq:adjoint_tau}. By saving these vectors, we can then take $z_n$, feed forward through $f(z_n, \theta)$ and backpropagate using the vector jacobian product. This can be done in one gradient update, accumulating a gradient with respect to $\theta$ and updating it after going through all adjoints, or we can update weights in batches as it is common in neural network training. The latter is shown to be the more stable and faster converging of the methods (see Figure \ref{fig:comparebatched}). 

\begin{figure}
\centering
\begin{subfigure}{0.4\textwidth}
  \centering
  \includegraphics[width=1.0\linewidth]{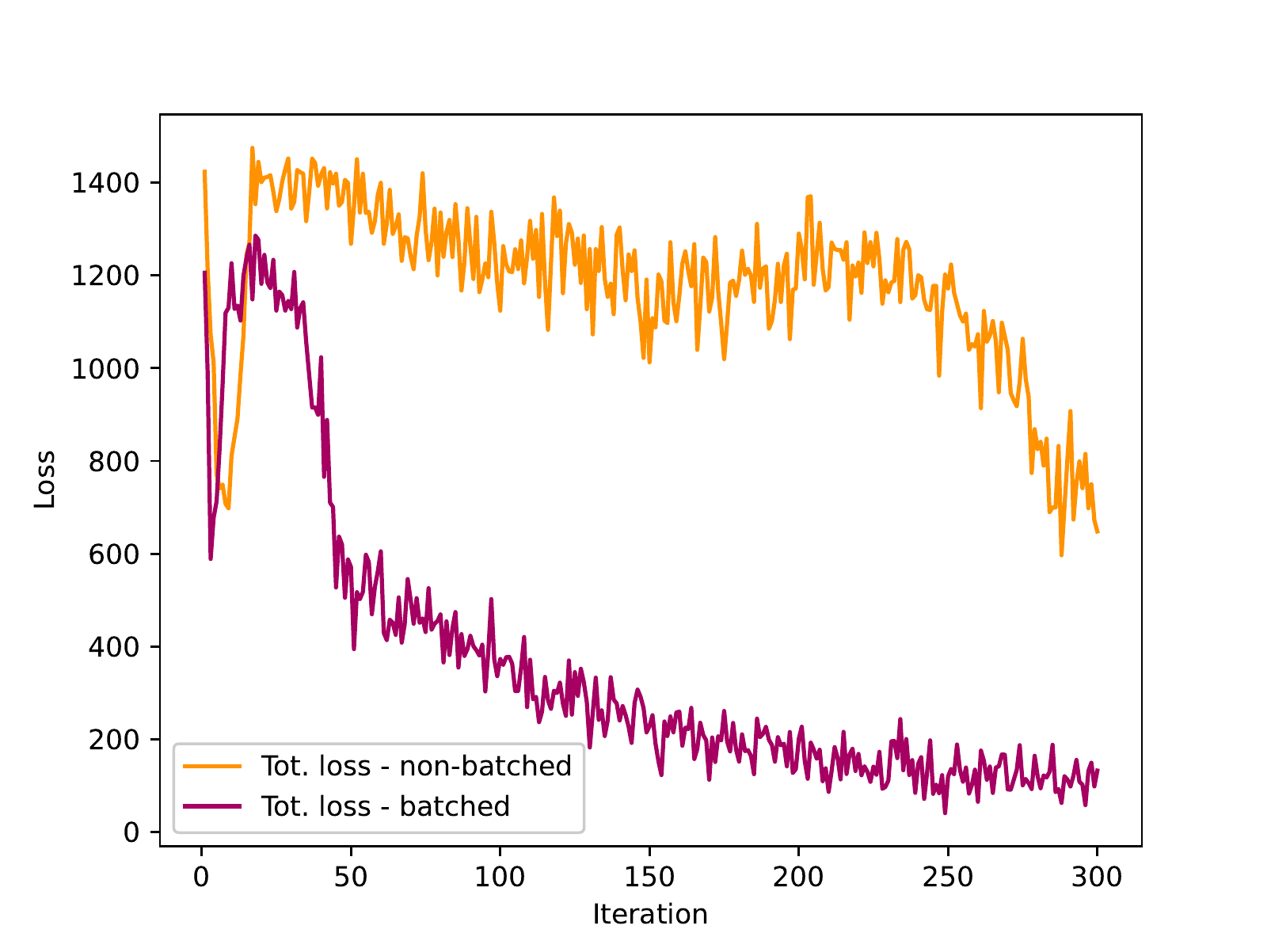}
\end{subfigure}%
\begin{subfigure}{0.4\textwidth}
  \centering
  \includegraphics[width=1.0\linewidth]{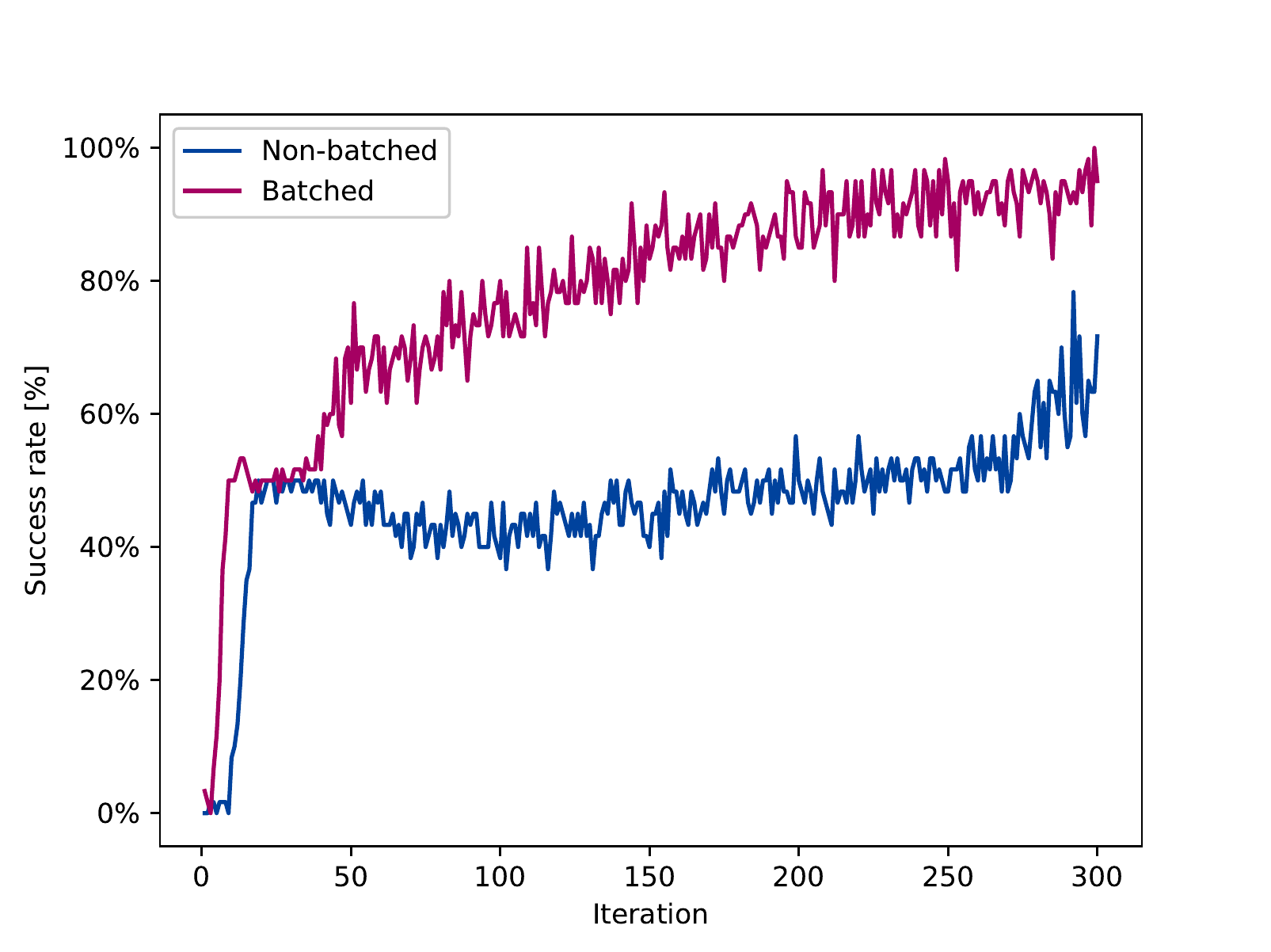}
\end{subfigure}
\caption{Comparison of the batched and one-time update of the weights in the 2D example from \ref{sec:2dexample}. The learning rate for the unbatched example was set approximately a number of batches times larger than for the batched run. The convergence is clearly more stable and even faster in the batched case. This was also observed for any other setting we tried during the development.}
\label{fig:comparebatched}
\end{figure}

\section{2D Muller-Brown potential} \label{ape:mb}
The equation of the PES:
\begin{equation}
    U_\mathrm{MB}(x,y) = B \sum_{i=1}^4 A_i \operatorname{exp}\left[\alpha_i (x-x_0)^2 + \beta_i (x-x_0)(y-y_0) + \gamma_i (y-y_0)^2 \right]
\end{equation}
The parameters used in this work are:
\begin{center}
\begin{tabular}{||c c c c c c c||} 
 \hline
 $i$ & $A_i$ & $\alpha_i$ & $\beta_i$ & $\gamma_i$ & $x_0$ & $y_0$ \\ [0.5ex] 
 \hline\hline
 1 & -1.73 & 0 & -0.39 & -3.91 & 48 & 8 \\ 
 \hline
 2 & -0.87 & 0 & -0.39 & -3.91 & 32 & 16 \\
 \hline
 3 & -1.47 & 4.3 & -2.54 & -2.54 & 24 & 31 \\
 \hline
 4 & 0.13 & 0.23 & 0.273 & 0.273 & 16 & 24 \\
 \hline
\end{tabular}
\end{center}
The barrier parameter $B=10 \ \text{kcal/mol}$

\section{5D Generalization} \label{ape:general}

We consider a generalization of the Muller-Brown potential. 
By adding three harmonic DoFs we complicate the problem and make it necessary to use a general form of a biasing potential, dependent on all degrees of freedom, as we do not know which of them defines the reaction.
The resulting potential has the form:
\begin{equation}
    U_{5D}(x_1, x_2, x_3, x_4, x_5) = U_\mathrm{MB}(x_1, x_3) + \kappa (x_2^2+x_4^2+x_5^2)
\end{equation}
The parameters for $U_\mathrm{MB}(x_1, x_3)$ are identical to the 2D-case, the new parameter  $\kappa = 0.1$
The results for the 5D case are visualized in the figure \cref{fig:2dresults}.

\begin{figure}
\begin{subfigure}[b]{0.3\linewidth}
    \centering
  \includegraphics[width=1.0\linewidth]{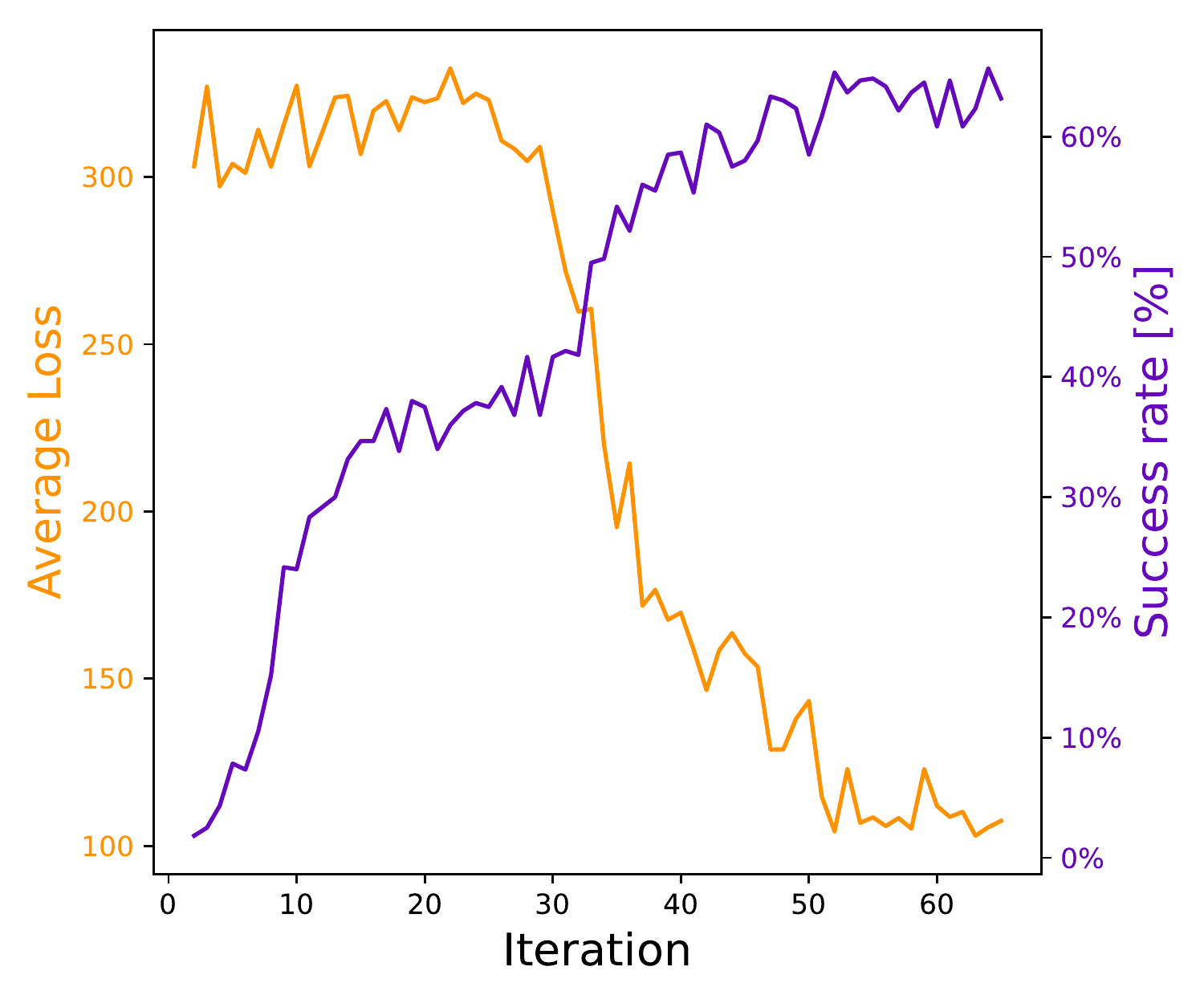}
  \phantomsubcaption
  \label{fig:5dsucrate}
\end{subfigure}
\centering
  \begin{subfigure}[b]{0.3\linewidth}
  \centering
  \includegraphics[width=1.0\linewidth]{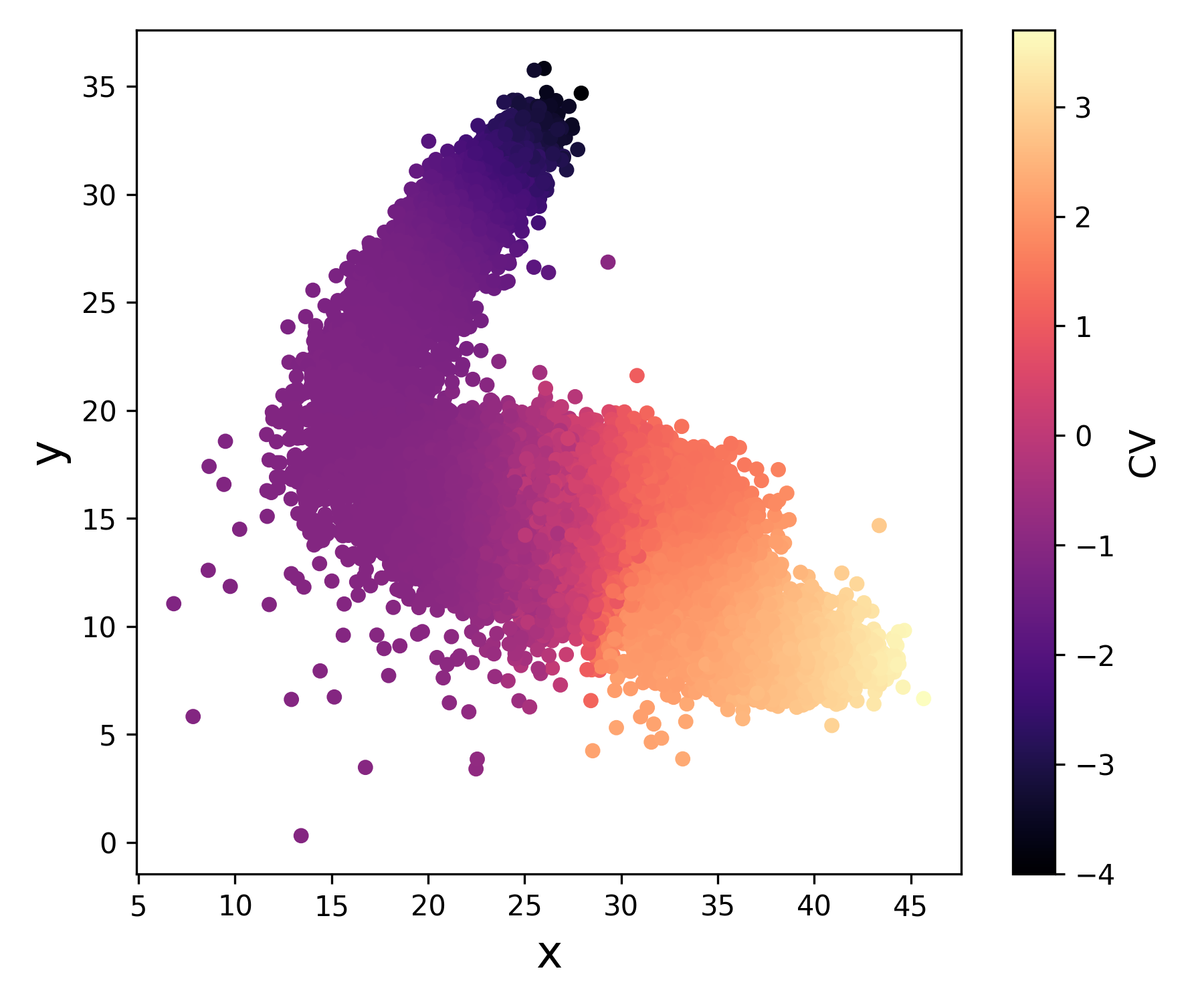}
  \phantomsubcaption
  \label{fig:5dcvs}
\end{subfigure}%
  \begin{subfigure}[b]{0.3\linewidth}
  \centering
  \includegraphics[width=1.0\linewidth]{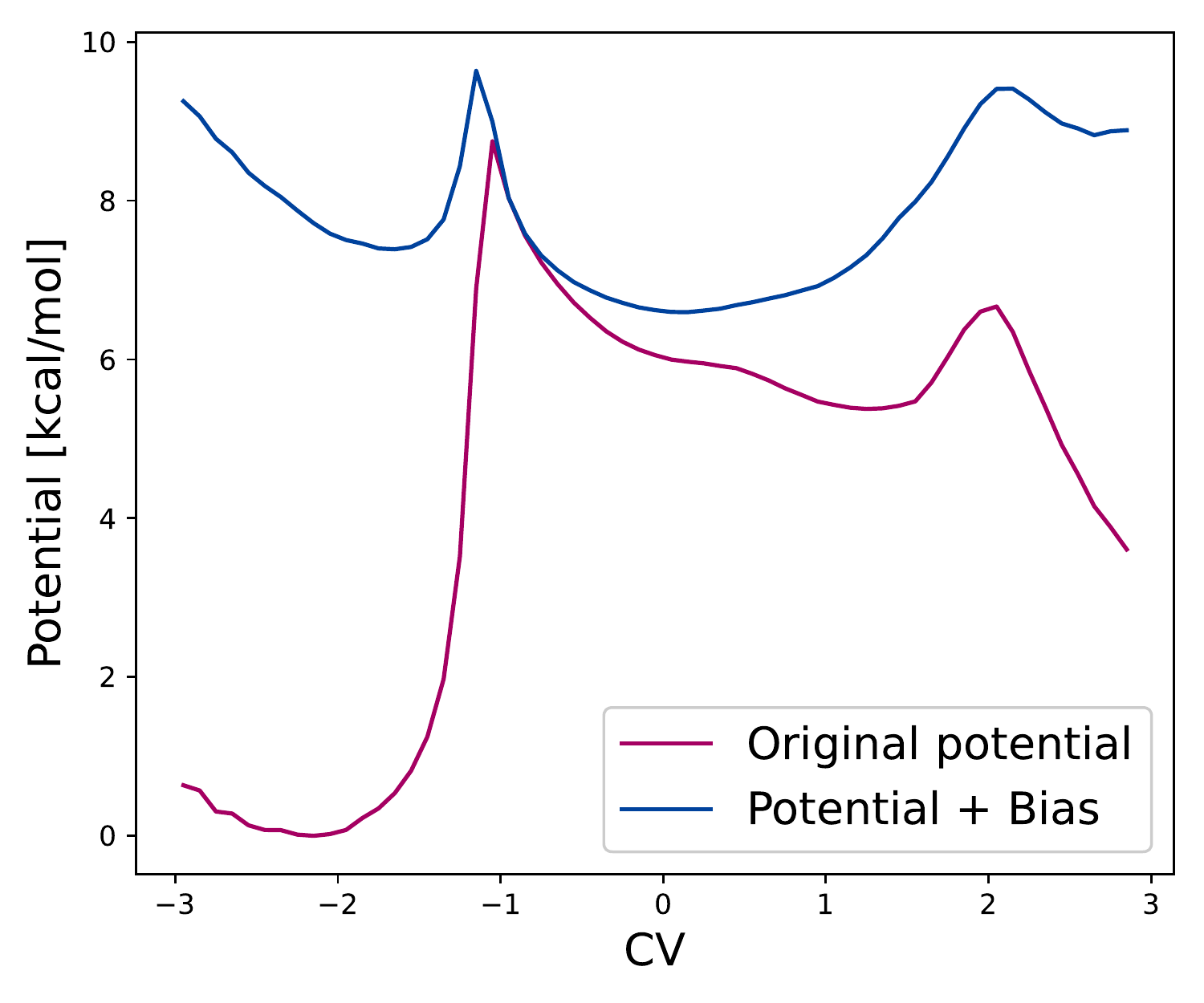}
  \phantomsubcaption
  \label{fig:5dpot}
\end{subfigure}
\vspace*{-5mm}
\caption{Results for the 5D extension of the Muller-Brown potential. 
\emph{left:} Evolution of loss value and probability of barrier crossing during the training progresses. 
\emph{middle:} CV determined with a Variational Autoencoder trained on a fully diffusive trajectory. The collective variables are not sharp around the transition region due to the high variance of the other noisy DoFs. 
This could be improved by more data, and more refined dimensionality reduction techniques that include temporal data such as e.g TiCA \cite{Schwantes2015ModelingTrick} or time-lagged autoencoders \cite{Wehmeyer2018Time-laggedKinetics}. \emph{right:} Average potential energy along the VAE collective variable with and without bias. The barriers were lowered to the level where they could be crossed with high probability.}
\label{fig:5dresults}
\end{figure}

\section{Alanine dipeptide simulation settings} \label{ape:alaparams}

The PMF of alanine dipeptide was obtained by means of well-tempered metadynamics\cite{Barducci2008Well-TemperedMethod}.
The deposited Gaussians had an initial height of 1~kcal/mol, a width of 10$^\circ$ along both dihedrals, and were deposited every 50~fs.
The WTMetaD temperature was 4000~K a
The simulation time step was 1~fs and the temperature was kept at 300~K with the Langevin thermostat with a friction constant of 1~ps$^{-1}$.
The total simulation time was 50~ns.

\section{Differentiable simulation parameters} \label{ape:diffsimparams}

The equations we simulate are \eqref{eq:biased_dyn}, discretized by the Leapfrog algorithm. The method is symplectic and conserves energy. The constants and parameters of the method were chosen as follows:
\begin{center}
\begin{tabular}{||c | c c c c c c||} 
 \hline
 case & $m$ [g/mol] & $\gamma$ [ps$^{-1}$] & $T$ [K] & d$t$ [fs] & timesteps & epochs \\ [0.5ex] 
 \hline
 2D & 0.1 & 0.1 & 10 & 1 & $6 \times 10^3$ & 101\\
 5D & 0.01 & 1.0 & 300 & 1 & $2 \times 10^4$ & 66\\
 Ala2 & - & 0.1 & 300 & 1 & $10^4$ & 301\\
 \hline
\end{tabular}
\end{center}
The column "timesteps" lists for how many steps we propagate a single simulation in each epoch. 
The update of parameters then represents an epoch. 
For the backward dynamics, we use 190 adjoints directly before the point where the loss function is calculated. 

\textbf{Batches of Parallel MDs}
These batches refer to the number of replicas that are simulated simultaneously. Using GPUs, we can parallelize the computation of forces and time step integration and thus are able to run 600 systems at once with a  similar speed of running just one. Accumulating the simulated data from so many systems allows us to increase the number of adjoints obtained and enables us to use a lower learning rate, making the training more stable. 

The setup of the bias function differed for every test case:

\textbf{2D Muller-Brown:} 
In this case, we use the setup described in \eqref{eq:gaussbias} with 50 times 50 basis functions. 

\textbf{5d Muller-brown:}
We use a fully connected network with all five degrees of freedom used as five continuous input neurons. The network has four hidden layers, each 150 neurons with SiLU as activation functions. The final layer has a single output neuron - the bias -  and no activation. 

\textbf{Alanine dipeptide:}
In the case of real molecules, the bias function gets more complicated. We define the Gaussian basis set in every single degree of freedom represented by a basis vector $\mathbf{e}_j$ and form a vector
\begin{equation}
    \bvv(\xx) = \sum_{j=1}^{n_{dof}} \sum_{i=1}^{n_g}  \operatorname{exp} \left( - \frac{(\xx-\xx_{ij}^0)^2}{2\sigma^2} \right) \mathbf{e}_j.
\end{equation}
$n_{dof}$ represents the total number of candidate CVs or degrees of freedom considered. In our case, this was 5. $n_g$ is the total number of basis functions defined separately for every candidate CV. In our case, this was 50 and since we described dihedral angles, centers $\xx_{ij}^0$ were distributed uniformly from 0 to $2\pi$, respecting the periodicity of dihedrals. 
The flattened vector $\bvv(\xx)$ with size $n_{dof} \cdot n_g$ is then used as an input to a fully connected neural network with three hidden layers, each 150 neurons with SiLU as an activation function. The final layer has one output neuron without an activation function. 

As reactant and product we use the minima $\beta$ ($\phi \approx -2,\ \psi \approx 2$)  and C$7_\mathrm{ax}$ ($\phi \approx 1,\ \psi \approx -1.5$), respectively (compare right panel in \cref{fig:compareala}).
As candidate DoFs, we choose the four dihedral angles along the backbone of alanine dipeptide, denoted as $\theta_1$, $\phi$, $\psi$, $\theta_2$ in the Figure 1 from the paper \cite{Mironov2019ADipeptide}. To make thinks more complicated, we add one dihedral involving the side-chain methyl group that is expected to be correlated with $\phi$, defined by atoms $C_1$, $N_1$, $C_{\alpha}$, $C_{\beta}$ using the notation from the same Figure. 
In each dihedral angle, a Gaussian basis set accounting for periodicity is defined \cref{ape:alaparams}. 
These expanded dihedrals are then input to a fully connected network that calculates the bias function. 
The detailed settings are listed in  \cref{ape:diffsimparams}. 
As $\phi$ and $\psi$ are known, the results are reported as Ramachandran plots. 
No additional dimensionality reduction via VAEs is performed in this example.

The numerical tool used for the DiffSim of alanine dipeptide was partially based on components from the TorchMD library \cite{Doerr2021TorchMD:Simulations}. 
The simulations were carried out with the Amber ff19SB forcefield \cite{Tian2020Ff19SB:Solution} in vacuum.

\textbf{Graph-Minibatching batch size}
This batch size refers to the mini-batching of the computational graph illustrated in \cref{ape:graph_minibatch}. 
Here we split the accumulated adjoints into smaller batches and train the network sequentially. 
The mini-batch of 120 was used for all systems. We use the learning rate as a learning factor divided by the number of replicas to make it independent of the number of systems simulated simultaneously. The learning factor is chosen as 20 for the 2D case, 6 for the 5D case and 3 for the Alanine dipeptide. This, with 300 replicas running from reactant to the product and 300 the other way, gives us learning rates on the orders $10^{-2}$ to $10^{-3}$. We use Adam optimizer for all our cases. 

For the CV construction via Variational Autoencoder a simple setup was employed with a two hidden layer encoder and a two hidden layer decoder with 50 neurons and a Softplus activation function for each hidden layer. 

\section{Brachistochrone curve}
\label{ape:brachisto}

\begin{figure}
\begin{subfigure}[b]{0.4\linewidth}
    \centering
  \includegraphics[width=1.0\linewidth]{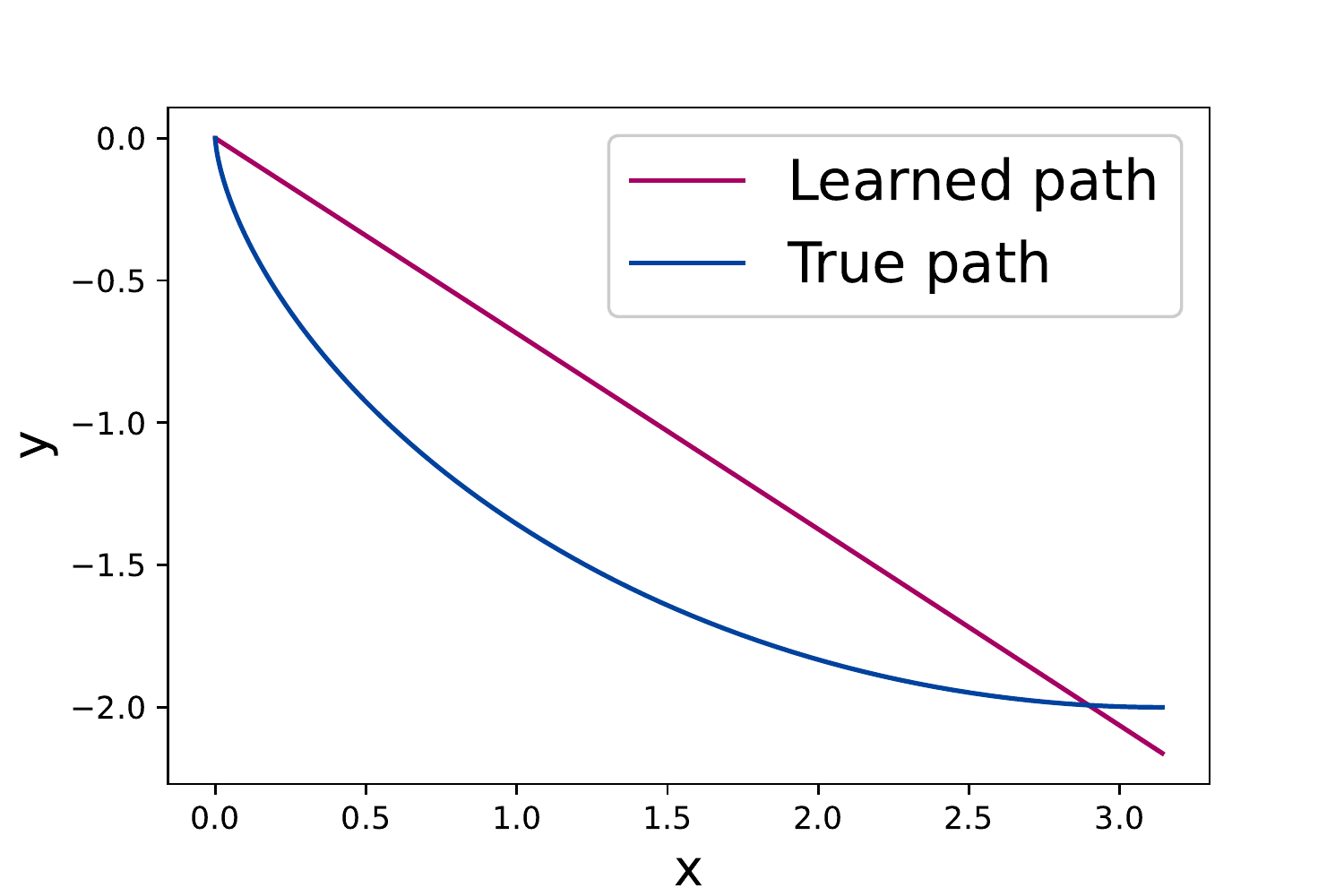}
  \phantomsubcaption
  \label{fig:bch0}
\end{subfigure}
\centering
  \begin{subfigure}[b]{0.4\linewidth}
  \centering
  \includegraphics[width=1.0\linewidth]{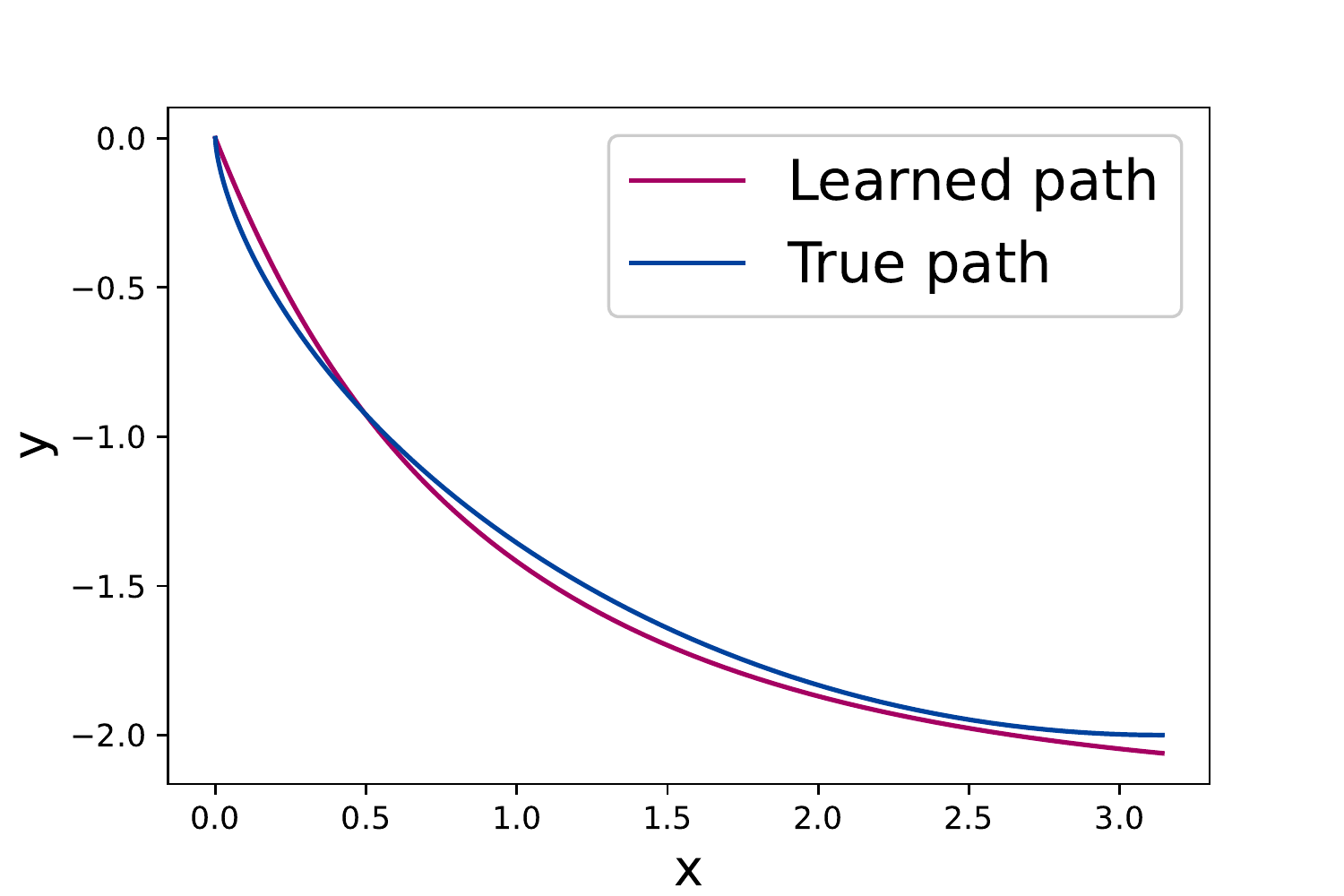}
  \phantomsubcaption
  \label{fig:bch200}
\end{subfigure}%
\vspace*{-5mm}
\caption{\emph{left} Initial state. The path is initialized as an almost straight path. \emph{right} After 200 iterations of differentiable simulations training, the path approximates the true path. The difference between the true curve and the one obtained by training is likely in the numerical scheme used to evaluate the integral.}
\label{fig:bchResults}
\end{figure}

Here we exemplify how one can employ differentiable simulations and their capabilities to optimize path dependent integrals and solve the Brachistochrone problem. The problem is formulated as follows:
Given a mass freely sliding on a curve $y = y(x)$  in the gravitational field $g$, find the curve from point \textbf{A} to lower point \textbf{B} for which the sliding time is the shortest. We assume no friction or air resistance and assume that \textbf{B} does not lie directly below \textbf{A}. 
For simplicity, we choose \textbf{A} to be the origin of the coordinate system.
The solution, the cyclone curve, of this famous problem was obtained by Leibniz, L'Hospital, Newton, and Bernoulli brothers \cite{Boyer1991AEdition}. A modified version, where we allow for an arbitrary difference in height between the two points and only prescribe their horizontal distance $\Delta x$ was solved by Lagrange and much later summarized and written in the modern language of variational formalism by \cite{Mertens2008BrachistochronesEnds}. In this case, a solution is also a cyclone with some parameters fixed. We prescribe the horizontal $\Delta x$ to be $\pi$ and search for a solution using differentiable simulations. A simple fully connected neural network $f(x)$ serves as a derivative of the curve $f(x) = \frac{d y(x)}{d x}$, so that $y(x)$ is then obtained by the path integration of the neural network. After integration, we numerically evaluate the time from the simulated path
\begin{equation}
    t = \int_0^l \frac{ds}{v(x)} = \int_0^\pi \sqrt{\frac{1+\left(\frac{d y(x)}{d x}\right)^2}{-2gy(x)}} dx
\end{equation}
and minimize it. In the first integrat, $l$ represents the length of the curve. The formula can be easily derived from the conservation of kinetic energy and from a Pythagorean expression $ds^2 = dx^2 + dy^2$. 
For the path construction and backpropagation we employ the \emph{torchdiffeq} python package shipped with the paper \cite{Chen2018NeuralEquations}. 

In this example, we present a problem that could not be solved with just a point-wise neural network optimization but requires consideration of a full path. The results for a short training are in \cref{fig:bchResults}.

\section{Data availability}

The source code for all examples (2D and 5D Muller-Brown, Alanine Dipeptide) are available online on Github \url{https://github.com/martinsipka/rarediffsim}

The Google Colab notebook with the Brachistochrone example is available at: \url{https://colab.research.google.com/drive/1YjIMTFQA0L9oLMkNpV7E2jOxbbzO6PnM?usp=sharing}

\end{document}